\documentclass[runningheads,a4paper,orivec,envcountsame,envcountsect,draft]{llncs}

\usepackage{ifdraft}
\ifdraft{
\usepackage[marginparwidth=2cm,nohead,nofoot,
            headsep = 1cm,
            footskip = 1cm,
            text={12.2cm,19.3cm},
            paperwidth=16.2cm,
            paperheight=23.3cm,
            marginparsep=1mm,
            marginparwidth=1.8cm,
            footskip=1cm,
            centering
            ]{geometry}
}{}

\usepackage{mathrsfs}
\usepackage{xcolor}
\usepackage{seqsplit}
\usepackage{xstring}

\usepackage[utf8]{inputenc}
\usepackage[T1]{fontenc}
\usepackage[british]{babel}
\usepackage{xspace}
\usepackage{dsfont}
\usepackage{tabularx}

\usepackage[inline]{enumitem}
\setlist[enumerate,1]{label=(\arabic*),font=\normalfont,align=left,leftmargin=0pt,labelindent=0pt,listparindent=\parindent,labelwidth=0pt,itemindent=!,topsep=3pt,parsep=0pt,itemsep=3pt,start=1}
\setlist[enumerate,2]{label=(\alph*),font=\normalfont,labelindent=*,leftmargin=*,start=1}
\setlist[itemize]{labelindent=*,leftmargin=*,topsep=5pt,itemsep=3pt}
\setlist[description]{labelindent=*,leftmargin=*,itemindent=-1 em}

\usepackage[%
rm={oldstyle=false,proportional=true},%
sf={oldstyle=false,proportional=true},%
tt={oldstyle=false,proportional=true,variable=true},%
qt=false%
]{cfr-lm}

\usepackage[noadjust]{cite} 
\usepackage{graphicx,csquotes}
\usepackage[final]{hyperref}
\hypersetup{
  colorlinks=true,
  linkcolor=black,
  citecolor=black,
  filecolor=black,
  urlcolor=black,
  pdftitle={Nondeterministic Syntactic Complexity},
  pdfauthor={Robert Myers, Stefan Milius, Henning Urbat},
  pdfkeywords={nondeterministic automata, semilattices},
  pdfduplex={DuplexFlipLongEdge},
}
  
\ifdraft{%
  \makeatletter
  \setcounter{tocdepth}{2}
  \makeatother
}{}

\usepackage{amssymb,amsxtra,mathrsfs,mathtools,bm,stmaryrd}
\numberwithin{equation}{section}
\usepackage[all,2cell]{xy}
\UseAllTwocells
\SelectTips{cm}{}
\newdir{ >}{{}*!/-5pt/@{>}}
\newdir{ (}{{}*!/-10pt/\dir^{(}}
\newdir{ )}{{}*!/-10pt/\dir_{(}}

\usepackage[footnote,nomargin]{fixme} 
\FXRegisterAuthor{hu}{ahu}{HU} 
\FXRegisterAuthor{sm}{asm}{SM} 
\FXRegisterAuthor{rm}{arm}{RM} 

\usepackage{microtype}
\allowdisplaybreaks

\addto\extrasbritish{
}

\spnewtheorem{assumptions}[theorem]{Assumptions}{\bfseries}{\rmfamily}
\spnewtheorem{notation}[theorem]{Notation}{\bfseries}{\rmfamily}
\spnewtheorem{observation}[theorem]{Observation}{\bfseries}{\rmfamily}
\spnewtheorem{defn}[theorem]{Definition}{\bfseries}{\rmfamily}
\spnewtheorem{expl}[theorem]{Example}{\bfseries}{\rmfamily}
\spnewtheorem{rem}[theorem]{Remark}{\bfseries}{\rmfamily}
\spnewtheorem{construction}[theorem]{Construction}{\bfseries}{\rmfamily}
\spnewtheorem{examples}[theorem]{Examples}{\bfseries}{\rmfamily}
\spnewtheorem{example_}[theorem]{Example}{\bfseries}{\rmfamily}

\spnewtheorem*{HSP}{General HSP Theorem}{\bf}{\itshape}
\spnewtheorem*{CompThm}{General Completeness Theorem}{\bf}{\itshape}

\usepackage{etoolbox}
\makeatletter
\newenvironment{notheorembrackets}{%
\csdef{@spopargbegintheorem}##1##2##3##4##5{\trivlist%
      \item[\hskip\labelsep{##4##1\ ##2}]{##4{##3}\@thmcounterend\ }##5}%
    }{%
\csdef{@spopargbegintheorem}##1##2##3##4##5{\trivlist%
      \item[\hskip\labelsep{##4##1\ ##2}]{##4(##3)\@thmcounterend\ }##5}%
    }
\makeatother

\renewcommand{\dim}[1]{\mathrm{dim}(#1)}

\newcommand{\dash}{\mathord{-}}

\newcommand{\Set}{\mathbf{Set}}

\newcommand{\JSL}{\mathbf{JSL}}

\newcommand{\C}{\mathscr{C}}

\newcommand{\ol}{\overline}

\newcommand{\Init}[1]{\mathsf{Init}(#1)}
\newcommand{\Fin}[1]{\mathsf{Fin}(#1)}

\newcommand{\At}[1]{\mathrm{At}(#1)}

\renewcommand{\epsilon}{\varepsilon}

\newcommand{\id}{\mathit{id}}

\newcommand{\seq}{\subseteq}
\newcommand{\xra}{\xrightarrow}

\newcommand{\op}{\mathsf{op}}

\renewcommand{\o}{\circ}

\newcommand{\takeout}[1]{\empty}

\renewcommand{\phi}{\varphi}

\newcommand{\Lra}{\Leftrightarrow}

\newcommand{\PSPACE}{\mathrm{PSPACE}}
\newcommand{\NP}{\mathrm{NP}}

\newcommand{\Rel}{\mathbf{Rel}}

\newcommand{\under}[1]{|#1|}

\newcommand{\epito}{\twoheadrightarrow}

\newcommand{\down}{\downarrow}

\newcommand{\monoto}{\rightarrowtail}

\newcommand{\f}{\mathsf{f}}

\newcommand{\rev}[1]{{#1}^\mathsf{r}}

\newcommand{\Pow}{\mathcal{P}}

\newcommand{\LW}[1]{\mathsf{LD}(#1)}
\newcommand{\jslLQ}[1]{\mathsf{SLD}(#1)}
\newcommand{\jslBLQ}[1]{\mathsf{BLD}(#1)}
\newcommand{\jslBLRQ}[1]{\mathsf{BLRD}(#1)}
\newcommand{\xto}[1]{\xrightarrow{#1}}

\newcommand{\jslLangs}[1]{\mathsf{langs}(#1)}

\newcommand{\Aut}[1]{\mathbf{Aut}(#1)}

\newcommand{\Syn}[1]{\mathsf{syn}(#1)}

\renewcommand{\ts}[1]{\mathsf{ts}(#1)}
\newcommand{\tm}[1]{\mathsf{tm}(#1)}

\newcommand{\dr}{\mathrm{dr}}

\newcommand{\rqc}[1]{\mathsf{rdc}(#1)}
\newcommand{\ns}[1]{\mathrm{ns}(#1)}
\newcommand{\nsyn}[1]{\mathrm{n}\mu(#1)}

\newcommand{\Nat}{\mathds{N}}

\newcommand{\To}{\Rightarrow}

\makeatletter
\newcommand{\xTo}[2][]{\ext@arrow 0359\Rightarrowfill@{#1}{#2}}
\makeatother

\newcommand{\minDfa}[1]{\mathsf{dfa}(#1)}

\newcommand{\rsc}[1]{\mathsf{rsc}(#1)}
\newcommand{\reach}[1]{\mathsf{reach}(#1)}

\newcommand{\simple}[1]{\mathsf{simple}(#1)}

\newcommand{\rDR}[1]{\mathcal{DR}_{#1}}
\newcommand{\rRDR}[1]{\mathcal{DR}_{#1}^j}

\newcommand{\LQ}[1]{\mathsf{SLD}(#1)}

\newcommand{\LD}[1]{\mathsf{LD}(#1)}

\title{Nondeterministic Syntactic Complexity}
\titlerunning{Nondeterministic Syntactic Complexity}

\author{Robert S.~R. Myers \and Stefan Milius\inst{1}\fnmsep\thanks{Supported by Deutsche
    Forschungsgemeinschaft (DFG) under projects MI~717/5-2 and MI~717/7-1, and as part of the Research and Training Group 2475 ``Cybercrime and Forensic Computing'' (393541319/GRK2475/1-2019)} \and Henning Urbat\inst{1}\fnmsep\thanks{Supported by Deutsche
    Forschungsgemeinschaft (DFG) under proj. SCHR~1118/8-2}}
\authorrunning{R.~Myers, S.~Milius, and H.~Urbat}

\institute{$~^\text{1}$Friedrich-Alexander-Universit\"at Erlangen-N\"urnberg}

\begin{document}
\maketitle

\begin{abstract}
  We introduce a new measure on regular languages: their \emph{nondeterministic syntactic complexity}. It is the least degree of any extension of the `canonical boolean representation' of the syntactic monoid. Equivalently, it is the least number of states of any \emph{subatomic} nondeterministic acceptor.
  It turns out that essentially all previous structural work on nondeterministic state-minimality computes this measure. Our approach rests on an algebraic interpretation of nondeterministic finite automata as deterministic finite automata endowed with semilattice structure. Crucially, the latter form a self-dual category.
\end{abstract}

\section{Introduction}\label{S:intro}
Regular languages admit a plethora of equivalent representations: finite automata, finite monoids, regular expressions, formulas of monadic second-order logic, and numerous others. In many cases, the most succinct representation is given by a \emph{nondeterministic finite automaton (nfa)}. Therefore, the investigation of state-minimal nfas is of both computational and mathematical interest. However, this turns out to be surprisingly intricate; in fact, the task of minimizing an nfa, or even of deciding whether a given nfa is minimal, is known to be $\PSPACE$-complete~\cite{MinNfaHard1993}. One intuitive reason is that minimal nfas lack structure: a language may have many non-isomorphic minimal nondeterministic acceptors, and there are no clearly identified and easily verifiable mathematical properties distinguishing them from non-minimal ones. As a consequence, all known algorithms for nfa minimization (and related problems such as inclusion or universality testing) require some form of exhaustive search~\cite{KamedaWeiner1970,ddhr06,cm19}. This sharply contrasts the situation for minimal \emph{deterministic finite automata (dfa)}: they can be characterized by a universal property making them unique up to isomorphism, which immediately leads to efficient minimization.

In the present paper, we work towards the goal of bringing more structure into the theory of nondeterministic state-minimality. To this end, we propose a novel algebraic perspective on nfas resting on \emph{boolean representations} of monoids, i.e.~morphisms $M\to \JSL(S,S)$ from a monoid $M$ into the endomorphism monoid of a finite join-semilattice $S$. Our focus lies on quotient monoids of the free monoid $\Sigma^*$ recognizing a given regular language $L\seq\Sigma^*$. The largest such monoid is  $\Sigma^*$ itself, while the smallest one is the \emph{syntactic monoid} $\Syn{L}$. For both of them, $L$ induces a \emph{canonical boolean representation} 
\[\Sigma^*\to \JSL(\jslLQ{L}, \jslLQ{L} \qquad\text{and}\qquad \Syn{L}\to \JSL(\jslLQ{L},\jslLQ{L})\]
 on the semilattice $\jslLQ{L}$ of all finite unions of left derivatives of $L$. The first representation gives rise to an algebraic characterization of minimal nfas:

\medskip\noindent\textbf{Theorem.} The size of a state-minimal nfa for $L$ equals the least degree of any extension of the canonical representation of $\Sigma^*$ induced by $L$.

\medskip\noindent Here, the \emph{degree} of a representation refers to the number of join-irreducibles of the underlying semilattice. In the light of this result, it is natural to ask for an analogous automata-theoretic perspective on the canonical representation of $\Syn{L}$ and its extensions. For this purpose, we introduce the class of \emph{subatomic} nfas, a generalization of \emph{atomic} nfas earlier introduced by Brzozowski and Tamm~\cite{TheoryOfAtomataBrzTamm2014}. In order to get a handle on them, we employ an algebraic framework that interprets nfas in terms of \emph{$\JSL$-dfas}, i.e.~deterministic finite automata in the category of semilattices. In this setting, the semilattice $\LQ{L}$ used in the canonical representations naturally arises as the \emph{minimal} $\JSL$-dfa for the language $L$. We shall demonstrate that much of the structure theory of (sub-)atomic nfas reduces to the observation that the category of $\JSL$-dfas is \emph{self-dual}. Our main result gives an algebraic characterization of minimal subatomic nfas:

\medskip\noindent\textbf{Theorem.} The size of a state-minimal subatomic nfa for $L$ equals the least degree of any extension of the canonical representation of $\Syn{L}$.

\medskip\noindent We call the measure suggested by the above theorem the \emph{nondeterministic syntactic complexity} of the language $L$. It turns out to be extremely natural: as illustrated in \autoref{S:examples}, essentially all existing work on the structure of state-minimal nfas implicitly identifies classes of languages whose nondeterministic state complexity equals their nondeterministic syntactic complexity, and thus is actually concerned with computing minimal subatomic acceptors.

\section{Preliminaries}
We start by introducing some notation and terminology used in the paper.

\medskip\noindent\emph{Semilattices.} A \emph{(join-)semilattice} is a poset $(S,\leq_S)$ in which every finite subset $X\seq S$ has a least upper bound, a.k.a.\, join, denoted by $\mathop{\bigvee} X$. A \emph{morphism} of semilattices is a map preserving all finite joins. Let $\JSL$ denote the category of join-semilattices and their morphisms. An element $j$ of a semilattice $S$ is \emph{join-irreducible} if for all finite subsets $X\seq S$ with $j=\bigvee X$ one has $j\in X$. Let \[J(S) = \{\, j\in S \;:\;\text{$j$ is join-irreducible} \,\}.\] Let $2=\{0,1\}$ denote the two-element semilattice with $0\leq 1$. Since $2\cong (\Pow (1),\seq)$ is the free semilattice on a single generator, morphisms from $2$ into a semilattice $S$ correspond uniquely to elements of $S$. Similarly, a morphism $f\colon S \to 2$ corresponds uniquely to a \emph{prime filter} $F=f^{-1}[1]\seq S$, i.e.~an upwards closed subset such that $\bigvee X \in F$ implies $X\cap F\neq \emptyset$ for every finite subset $X\seq S$. If $S$ is finite, prime filters are precisely the sets $F=\{ s\in S: s\not\leq s_0 \}$ for $s_0\in S$. If $S$ is a subsemilattice of a semilattice $T$, every prime filter $F$ of $S$ can be extended to the prime filter $T\setminus(\mathop{\down}(S\setminus F))$ of $T$, where $\mathop{\down}X = \{\, t\in T\;:\; t\leq x \text{ for some $x\in X$} \,\}$ denotes the down-closure of a subset $X\seq T$. Equivalently, every morphism $f\colon S\to 2$ can be extended to a morphism $g\colon T\to 2$. In category-theoretic terminology, this means that the semilattice $2$ forms an injective object of $\JSL$. 

The category $\JSL_\f$ of finite semilattices is \emph{self-dual}~\cite{StoneSpaces}. The equivalence functor $\JSL_\f\xra{\simeq}\JSL_\f^{\op}$ sends a semilattice $S$ to its \emph{dual semilattice} $S^\op$ 
obtained by reversing the order, and a morphism $f\colon S\to T$ to the morphism $f^\ast\colon T^{\op}\to S^{\op}$ mapping $t\in T$ to the $\leq_S$-largest element $s\in S$ with $f(s)\leq_T t$. Note that $f$ is \emph{adjoint} to $f^\ast$: for $s\in S$ and $t\in T$ we have $f(s)\leq_T t$ iff $s\leq_S f^\ast(t)$.

\medskip\noindent \emph{Languages.} A \emph{language} is a subset $L$ of $\Sigma^*$, the set of finite words over an alphabet $\Sigma$. We let $\ol{L} = \Sigma^*\setminus L$ denote the \emph{complement} and $\rev{L} = \{\rev{w}: w\in L \}$ the \emph{reverse}, where $\rev{w} = a_n\ldots a_1$ for $w=a_1\ldots a_n$. The \emph{left derivatives}, \emph{right derivatives} and \emph{two-sided derivatives} of $L$ are, respectively, given by
$u^{-1}L = \{w\in \Sigma^* : uw\in L\}$, $Lv^{-1} = \{ w\in \Sigma^*: wv\in L  \}$ and $u^{-1}Lv^{-1} = \{ w\in \Sigma^*: uwv\in L \}$ for $u,v\in \Sigma^*$. More generally, for $U\seq \Sigma^*$ the language $U^{-1}L = \bigcup_{u\in U} u^{-1}L$ is called the \emph{left quotient} of $L$ w.r.t.~$U$. We define the following sets of languages generated by $L$:
\begin{itemize}
    \item $\LW{L} = \{  u^{-1} L : u \in \Sigma^* \}$, the set of all left derivatives of $L$;
    \item $\jslLQ{L}$, its closure under finite union;
    \item $\jslBLQ{L}$, its closure under all set-theoretic boolean operations;
    \item $\jslBLRQ{L}$, its closure under all boolean operations and right derivatives.
  \end{itemize}
In other words, $\jslLQ{L}$ is the $\cup$-semilattice of all left quotients of $L$, or equivalently, the $\cup$-subsemilattice of $\Pow(\Sigma^*)$ generated by all left derivatives. Moreover, $\jslBLQ{L}$ and $\jslBLRQ{L}$ form the boolean subalgebras of $\Pow(\Sigma^*)$ generated by all left derivatives and all two-sided derivatives, respectively.

\section{Duality Theory of Semilattice Automata}
In this section, we set up the algebraic framework in which nondeterministic automata can be studied. Since it involves considering several different types of automata, it is convenient to view them all as instances of a general categorical concept. For the rest of this paper, let $\Sigma$ denote a fixed finite input alphabet.

\begin{defn}\label{def:automaton}
Let $\C$ be a category and let $X,Y\in \C$ be two fixed objects.
An \emph{automaton} in $\C$ is a quadruple $(S,\delta,i,f)$ consisting of an object $S\in \C$ of \emph{states}, a family $\delta=(\delta_a\colon S\to S)_{a\in \Sigma}$ of morphisms representing \emph{transitions}, and two morphisms $i\colon X\to S$ and $f\colon S\to Y$ representing \emph{initial} and \emph{final} states (see the left-hand diagram below).
A \emph{morphism} between automata $(S,\delta,i,f)$ and $(S',\delta',i',f')$ is given by a morphism $h\colon S\to S'$ in $\C$ preserving transitions, initial states and final states, i.e.~making the right-hand diagram below commute for all $a\in \Sigma$:
\[
\vcenter{
\xymatrix{
X \ar[r]^{i} & S \ar@(ul,ur)^{\delta_a} \ar[r]^f & Y
}
}
\qquad
\xymatrix@R-1em{
X \ar[r]^{i} \ar[dr]_{i'} & S \ar[r]^{\delta_a} \ar[d]_h & S \ar[d]^h \ar[r]^f & Y \\
 & S' \ar[r]_{\delta_a'} & S' \ar[ur]_{f'} & 
}
\]
Let $\Aut{\C}$ denote the category of automata in $\C$ and their morphisms.
\end{defn}

\begin{notation}
We put $\delta_w := \delta_{a_n}\o \cdots \o \delta_{a_1}$ for $w=a_1\ldots a_n$ in $\Sigma^*$.
\end{notation}

\begin{expl}
\begin{enumerate}
\item An automaton $D=(S,\delta,i,f)$ in $\Set$, the category of sets and functions, with $X=1$ and $Y=2$, is precisely a classical \emph{deterministic automaton}. It is called a $\emph{dfa}$ if $S$ is finite. We identify the map $i\colon 1\to S$ with an initial state $s_0=i(\ast)\in S$, and the map $f\colon S\to 2$ with a set $F=f^{-1}[1]\seq S$ of final states. The {language} $L(D,s)$ \emph{accepted} by a state $s\in S$ is the set of all words $w\in \Sigma^*$ such that $\delta_w(s)\in F$. The language $L(D)$ \emph{accepted} by $D$ is the language accepted by the state $s_0$.
\item An automaton $N=(S,\delta,i,f)$  in $\Rel$, the category of sets and relations, with $X=Y=1$, is precisely a classical \emph{nondeterministic automaton}. It is called an \emph{nfa} if $S$ is finite. We identify $i\seq 1\times S$ with a set $I\seq S$ of initial states and $f\seq S\times 1$ with a set $F\seq S$ of final states. Thus, in our view an nfa may have multiple initial states. The {language} $L(N,R)$ \emph{accepted} by a subset $R\seq S$ consists of all $w\in \Sigma^*$ such that $(r,s)\in \delta_w$ for some $r\in R$ and $s\in F$. The {language} $L(N)$ \emph{accepted} by $N$ is the language accepted by the set $I$.
\item  An automaton $A=(S,\delta,i,f)$ in $\JSL$ with $X=Y=2$, shortly a \emph{$\JSL$-automaton}, is given by a semilattice $S$ of states, a family $\delta = (\delta_a\colon S\to S)_{a\in \Sigma}$ of semilattice morphisms specifying  transitions, an initial state $s_0\in S$ (corresponding to $i\colon 2\to S$), and a prime filter $F\seq S$ of final states (corresponding to $f\colon S\to 2$). It is called a \emph{$\JSL$-dfa} if $S$ is finite. The language \emph{accepted} by a state $s\in S$ or by the automaton $A$, resp., is defined as for deterministic automata.
\end{enumerate}
\end{expl}

\begin{rem}[$\JSL$-dfas vs. nfas]\label{rem:jsldfa_vs_nfa}
Dfas, nfas and $\JSL$-dfas are expressively equivalent; they all accept precisely the regular languages. The interest of $\JSL$-dfas is that they constitute an algebraic representation of nfas:
\begin{enumerate}
\item Every $\JSL$-dfa $A=(S,\delta,s_0,F)$ induces an equivalent nfa $J(A)$ on the set $J(S)$ of join-irreducibles of $S$. Given $s,t\in J(S)$ and $a\in \Sigma$, there is a transition $s\xrightarrow{a} t$ in $J(A)$ iff $t\leq \delta_a(s)$; the initial states are those $s\in J(S)$ with $s\leq s_0$, and the  final states form the set $J(S)\cap F$. 
\item Conversely, for every nfa $N=(Q,\delta,I,F)$, the \emph{subset construction} yields an equivalent $\JSL$-dfa $\Pow(N)$ with states $\Pow(Q)$ (the $\cup$-semilattice of subsets of $Q$), transitions $\Pow{\delta_a}\colon \Pow(Q)\to \Pow(Q)$, $X\mapsto \delta_a[X]$, initial state $I\in \Pow(Q)$, and final states those subsets of $Q$ containing some state from $F$. Note that $J(\Pow(Q))\cong {Q}$.
\end{enumerate} 
 It follows that the task of finding a state-minimal nfa for a given language is equivalent to finding a $\JSL$-dfa with a minimum number of join-irreducibles~\cite{arbib_manes_1975}. This idea has recently been extended to a general coalgebraic framework \cite{mamu15,vhmss19}.
\end{rem}
Recall that the \emph{minimal dfa}~\cite{BrzozowskiDRE1964} for a regular language $L$, denoted by $\minDfa{L}$, has states $\LW{L}$ (the set of left derivatives of $L$), transitions $K\xra{a} a^{-1}K$ for $K\in \LW{L}$ and $a\in \Sigma$, initial state $L=\epsilon^{-1}L$, and final states those $K\in \LW{L}$ containing $\epsilon$. Up to isomorphism, it can be characterized as the unique dfa accepting $L$ that is \emph{reachable} (i.e.~every state is reachable from the initial state via transitions) and \emph{simple} (i.e.~any two distinct states accept distinct languages). We now develop the analogous concepts for $\JSL$-automata; they are instances of the categorical theory of minimality due to Arbib and Manes~\cite{am75} and Goguen~\cite{goguen75}. Let us first observe that every language has two canonical infinite $\JSL$-acceptors:

\begin{defn} Let $L\seq\Sigma^*$ be a language.
\begin{enumerate}
\item The \emph{initial $\JSL$-automaton} $\Init{L}$ for $L$ has states $\Pow_\f(\Sigma^*)$ (the $\cup$-semilattice of finite subsets of $\Sigma^*$), initial state $\{\epsilon\}$, final states all $X\in \Pow_\f(\Sigma^*)$ with $X\cap L\neq \emptyset$, and transitions $X\mapsto Xa = \{xa\;:\; x\in X\}$ for $X\in \Pow_\f (\Sigma^*)$ and  $a\in \Sigma$.
\item The \emph{final $\JSL$-automaton} $\Fin{L}$ for $L$ has states $\Pow(\Sigma^*)$ (the $\cup$-semilattice of all languages),  initial state $L$, final states all languages $K$ containing $\epsilon$, and transitions $K\mapsto a^{-1}K$ for $K\in \Pow(\Sigma^*)$ and $a\in \Sigma$.
\end{enumerate}
\end{defn}
As suggested by the terminology, these automata form the initial and the final object in the category of $\JSL$-automata accepting $L$:

\begin{notheorembrackets}
\begin{lemma}[{\cite{goguen75,am75}}]\label{lem:initfinal}
 For every $\JSL$-automaton $A=(S,\delta,s_0,F)$ accepting the language $L\seq \Sigma^*$, there exist unique $\JSL$-automata morphisms
\[ e_A\colon \Init{L} \to A \qquad\text{and}\qquad m_A\colon A\to \Fin{L}. \]
The map $e_A$ sends $\{w_1,\ldots, w_n\}\in \Pow_\f(\Sigma^*)$ to the state $\bigvee_{i=1}^n \delta_{w_i}(s_0)$, and the map $m_A$ sends a state $s\in S$ to $L(A,s)$, the language accepted by $s$.
\end{lemma}
\end{notheorembrackets}

\begin{defn}
A $\JSL$-automaton $A=(S,\delta,s_0,F)$ is called 
\begin{enumerate}
\item \emph{reachable} if the unique morphism $e_A\colon \Init{L}\to A$ is surjective, i.e.~every state is of the form $\bigvee_{i=1}^n \delta_{w_i}(s_0)$ for some $w_1,\ldots, w_n\in \Sigma^*$;
\item \emph{simple} if the unique morphism $m_A\colon A \to \Fin{L}$ in injective, i.e. any two distinct states accept distinct languages;
\item \emph{minimal} if it is both reachable and simple.
\end{enumerate}
\end{defn}

\begin{rem}
\begin{enumerate}
\item The category $\Aut{\JSL}$ has a factorization system given by
  surjective and injective morphisms. Thus, for every $\JSL$-automata
  morphism $h\colon (S,\delta,i,f)\to (S',\delta',i',f')$ with image
  factorization $h=(\xymatrix@1@C-0.5em{S \ar@{->>}[r]^-e & S'' \ar@{ >->}[r]^-m & S'})$ in $\JSL$, there exists a unique $\JSL$-automaton structure $(S'',\delta'',i'',f'')$ on $S''$ making  both $e$ and $m$ automata morphisms. We call $e$ the \emph{coimage} and $m$ the \emph{image} of $h$. \emph{Subautomata} and \emph{quotient automata} of $\JSL$-automata are represented by injective and surjective morphisms, respectively.
\item Every $\JSL$-automaton $A$ has a unique reachable subautomaton $\reach{A}\monoto A$,  the \emph{reachable part} of $A$. It is the smallest subautomaton of $A$ and arises as the image of the unique morphism $e_A\colon \Init{L}\to A$. Thus,
\[\text{$A$ is reachable}\quad\text{iff}\quad A\cong \reach{A}\quad \text{iff}\quad \text{$A$ has no proper subautomaton}.\]
Let us emphasize that a state in $\reach{A}$ is not necessarily reachable when $A$ is viewed as an ordinary dfa. For distinction, we thus call a state \emph{$\JSL$-reachable} if it lies in $\reach{A}$, and \emph{dfa-reachable} if it is reachable in the usual sense.
\item Dually, every $\JSL$-automaton $A$ has a unique simple quotient automaton $A\epito \simple{A}$,  the \emph{simplification} of $A$. It is the smallest quotient automaton of $A$ and arises as the coimage of the unique morphism $m_A\colon A\to \Fin{L}$. Thus,
\[\text{$A$ is simple}\quad\text{iff}\quad A\cong \simple{A}\quad \text{iff}\quad \text{$A$ has no proper quotient automaton}.\]
\item Every language $L\seq \Sigma^*$ has a minimal $\JSL$-automaton, unique up to isomorphism. It can be constructed as the image of the unique automata morphism $h_L\colon \Init{L}\to \Fin{L}$. Since $h_L$ sends $\{w_1,\ldots,w_n\}\in \Pow_\f(\Sigma^*)$ to the language $\bigcup_{i=1}^n w_{i}^{-1}L$, the minimal automaton of $L$ is the subautomaton $\LQ{L}$ of $\Fin{L}$ carried by the semilattice of finite unions of left derivatives of $L$.
\end{enumerate}
\end{rem}

\begin{expl} The minimal $\JSL$-dfa accepting $L = \{ a, aa \}$ is shown below, with the dashed lines representing the partial order.
 {\tiny
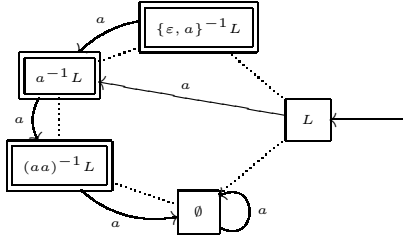
 \[
   \xymatrix@R-3.5em@C-2em{
   & *++[F=]{\{ \epsilon, a \}^{-1} L} \ar@/_10pt/[dl]_-a
   \\
   *++[F=]{a^{-1} L} \ar@{..}[ur] \ar@/_10pt/[dd]_a && 
   \\
   && *++[F-]{L} \ar@{..}[uul] \ar[ull]_a && \ar[ll]
   \\
   *++[F=]{(aa)^{-1} L} \ar@{..}[uu] \ar@/_10pt/[dr]_a  && 
   \\
   & *++[F-]{\emptyset} \ar@{..}[ul] \ar@{..}[uur] \ar@(dr,ur)_a
   }
 \]}
\end{expl}

\begin{rem} The self-duality of $\JSL_\f$ lifts to a self-duality of the category of $\JSL$-dfas. The equivalence functor $\Aut{\JSL_\f}\xra{\simeq}\Aut{\JSL_\f}^\op$ maps a $\JSL$-dfa $A= (S,\,(\delta_a\colon S\to S)_{a\in \Sigma},\, i\colon 2\to S, \,f\colon S\to 2)$ to its \emph{dual automaton}
\[ A^\op = (S^\op,\,(\delta_a^\ast\colon S^\op\to S^\op)_{a\in \Sigma},\, f^\ast\colon 2 \to S^\op,\, i^\ast\colon S^\op \to 2),\]
using that $2^\op\cong 2$. Thus, the initial state of $A^\op$ is the $\leq_S$-largest non-final state of $A$, and its final states are those $s\in S$ with $s_0\not\leq_S s$. Given $s,t\in S$ and $a\in \Sigma$, there is a transition $s\xra{a}t$ in $A^\op$ iff $t$ is the $\leq_S$-largest state with $\delta_a(t)\leq_S s$.
\end{rem}
The dualization of $\JSL$-dfas can be seen as an algebraic generalization of the reversal operation on nfas. Recall that the \emph{reverse} of an nfa $N$ is the nfa $\rev{N}$ obtained by flipping all transitions and swapping initial and final states. If $N$ accepts the language $L$, then $\rev{N}$ accepts the reverse language $\rev{L}$.

\begin{lemma}\label{lem:nfarev}
For each nfa $N=(Q,\delta,I,F)$, we have the $\JSL$-dfa isomorphism
\[
[\Pow(N)]^\op \xra{\cong} \Pow(\rev{N}), \qquad X\mapsto \ol{X}=Q\setminus X.
\]
\end{lemma}
The following lemma summarizes some important properties of $A^\op$:

\begin{lemma}\label{lem:astarprops} Let $A=(S,\delta,i,f)$ be a $\JSL$-dfa.
\begin{enumerate}
\item For every $s\in S$, we have $L(A^\op,s) = \{\,  w\in \Sigma^*\;:\; \delta_{\rev{w}}(s_0)\not\leq_S s \,\}$.
\item If $A$ accepts the language $L$, then $A^\op$ accepts the reverse language $\rev{L}$.
\item We have $[\reach{A}]^\op \cong \simple{A^\op}$. Thus, $A$ is reachable iff $A^\op$ is simple.
\end{enumerate}
\end{lemma}
Our next goal is to give, for every regular language $L$, dual characterizations of $\LQ{L}$, $\jslBLQ{L}$ and $\jslBLRQ{L}$, the $\JSL$-subautomata of $\Fin{L}$ carried by all finite unions of left derivatives, boolean combinations of left derivatives and boolean combinations of two-sided derivatives, respectively. These results form the core of our duality-based approach to (sub-)atomic nfas in the next section. The minimal $\JSL$-dfa $\LQ{L}$ admits the following dual description:
\begin{proposition}\label{prop:lqdual}
For every regular language $L$, the minimal $\JSL$-dfas for $L$ and $\rev{L}$ are dual. More precisely, we have the $\JSL$-dfa isomorphism
\[ \dr_L\colon [\LQ{\rev{L}}]^\op \xra{\cong} {\LQ{{L}}},\qquad K\mapsto  (\ol{\rev{K}})^{-1}{L}.\]
\end{proposition}

\begin{rem}\label{rem:conway}
\begin{enumerate}
\item The isomorphism $\dr_L$ induces a bijection between the \emph{left} and \emph{right factors} of $L$, i.e.\, the inclusion-maximal left/right solutions of $X \cdot Y \subseteq L$. Conway~\cite{Conway71} observed that the left and right factors are respectively $\{\overline{\rev{K}} :K \in \jslLQ{\rev{L}}\}$ and $\{\overline{K} : K \in \jslLQ{L} \}$ and that they biject. Backhouse~\cite{BACKHOUSE2016824} observed that they are dually isomorphic posets. \autoref{prop:lqdual} provides an explicit automata-theoretic lattice isomorphism arising canonically via duality.
\item The isomorphism $\dr_L$ is tightly connected to the \emph{dependency relation}~\cite{HROMKOVIC2002202, NDLowBoundsHard2006} of a regular language $L$, i.e.~the binary relation given by
  \[
    \rDR{L} \subseteq \LD{L} \times \LD{L^r},
    \qquad
    \rDR{L} (u^{-1} L, v^{-1} L^r) \mathrel{\mathord{:\!\iff}} uv^r \in L.
  \]
  Its restriction $\rRDR{L} := \rDR{L} \cap J(\jslLQ{L}) \times J(\jslLQ{L^r})$ to the $\cup$-irreducible left derivatives of $L$ and $\rev{L}$ is called the \emph{reduced dependency relation}. The following theorem shows that the semilattice of left quotients and the dependency relation are essentially the same concepts. In part (3), we use that the isomorphism $\dr_L$ restricts to a bijection between the $\cup$-irreducible derivatives of $\rev{L}$ and the meet-irreducible elements of the lattice $\jslLQ{{L}}$.
\end{enumerate}
\end{rem}

\begin{theorem}[Dependency theorem]\label{thm:dependency}
  \label{thm:dep_thm}
  \begin{enumerate}
    \item We have the $\JSL$-isomorphism
    \[
      \jslLQ{L} \xra{\cong} (\{ \rDR{L}[X] : X \subseteq \LD{L} \}, \cup, \emptyset),
      \qquad
      K\mapsto \{ v^{-1} \rev{L} : v \in \rev{K} \,\}.
    \]
Note that its codomain forms a subsemilattice of $\Pow(\LW{\rev{L}})$.
\item For all $u,v\in \Sigma^*$ we have $\rDR{L}(u^{-1} L, v^{-1} L^r) \iff u^{-1} L \nsubseteq \dr_L(v^{-1} L^r)$. 
\item The following diagram in $\Rel$ commutes:
    \[
      \xymatrix@=12pt{
        J(\jslLQ{\rev{L}})  \ar[rr]_{\cong}^{\dr_L}
        && M(\jslLQ{{L}}) 
        \\
        J(\jslLQ{L}) \ar[u]^{\rRDR{L}} \ar@{=}[rr]
        && J(\jslLQ{L}) \ar[u]_{\nsubseteq} 
      }
    \]
  \end{enumerate}
\end{theorem}
\clearpage \noindent Let us now turn to a dual characterization of the $\JSL$-dfa $\jslBLQ{L}$:

\begin{proposition}\label{prop:blqdual}
For every regular language $L$, the $\JSL$-dfa $\jslBLQ{L}$ is dual to the subset construction of the minimal dfa for $\rev{L}$:
\[ [\jslBLQ{L}]^\op \;\cong\; \Pow(\minDfa{\rev{L}}). \]
The isomorphism maps $\{w_1^{-1}\rev{L},\ldots, w_n^{-1}\rev{L} \}\in \Pow(\minDfa{\rev{L}})$ to  $\bigcap_{i=1}^n \overline{\At{\rev{w}_i}}$, where $\At{x}$ is the unique atom (= join-irreducible) of $\jslBLQ{L}$ containing $x$.
\end{proposition}
To state the dual characterization of $\jslBLRQ{L}$, we recall two standard concepts from algebraic language theory~\cite{pin20}. The \emph{transition monoid} of a deterministic automaton $D=(S,\delta,i,f)$ is the image $\tm{D}\seq \Set(S,S)$ of the morphism
\[ \Sigma^* \to \Set(S,S),\quad w~\mapsto~ \delta_{w}.\] 
Thus, $\tm{M}$ is carried by the set of extended transition maps $\delta_w$ ($w\in\Sigma^*$) with multiplication given by $\delta_{v}\bullet \delta_w = \delta_{vw}$ and unit $\id_S=\delta_\epsilon\colon S\to S$.
We may view $\tm{D}$ as a deterministic automaton with initial state $\id_S$, final states all $\delta_w$ such that $w$ is accepted by $D$, and transitions $\delta_w\xra{a}\delta_{wa}$ for $w\in \Sigma^*$ and $a\in \Sigma$. This automaton accepts the same language as $D$.
The \emph{syntactic monoid} $\Syn{L}$ of a regular language $L\seq \Sigma^*$ is the transition monoid of its minimal dfa:
\[ \Syn{L} = \tm{\minDfa{L}}. \] Equivalently, $\Syn{L}$ is the quotient monoid of the free monoid $\Sigma^*$ modulo the \emph{syntactic congruence} of $L$, i.e~the monoid congruence on $\Sigma^*$ given by
\[ v \equiv_L w \quad\text{iff}\quad \forall x, y \in \Sigma^*: xvy \in L \iff xwy \in L.\] The associated surjective monoid morphism $\mu_L \colon \Sigma^* \epito \Syn{L}$, mapping $w\in \Sigma^*$ to its congruence class $[w]_L\in \Syn{L}$, is called the \emph{syntactic morphism}.

\begin{proposition}\label{prop:blrqdual}
For every regular language $L$, the $\JSL$-dfa $\jslBLRQ{L}$ is dual to the subset construction of $\Syn{\rev{L}}$, viewed as a dfa:
\[ [\jslBLRQ{L}]^\op \;\cong\; \Pow(\Syn{\rev{L}}). \]
The isomorphism maps $\{\,[w_1]_{\rev{L}},\ldots, [w_n]_{\rev{L}} \,\}\in \Pow(\Syn{\rev{L}})$ to  $\bigcap_{i=1}^n \overline{\At{\rev{w_i}}}$, with $\At{x}$ denoting the unique atom of $\jslBLRQ{L}$ containing $x$.
\end{proposition}
Our final duality result in this section concerns the \emph{transition semiring}~\cite{Polak2001}, a generalization of the transition monoid to $\JSL$-automata. Note that the monoid $\JSL(S,S)$ of endomorphisms of a semilattice $S$ forms an idempotent semiring with join defined pointwise: for any $f,g\colon S\to S$, the morphism $f\vee g\colon S\to S$ is given by $s\mapsto f(s)\vee g(s)$. The {transition semiring} of a $\JSL$-automaton $A=(S,\delta,i,f)$ is the image $\ts{A}\seq \JSL(S,S)$ of the semiring morphism
\[ \Pow_\f(\Sigma^*) \to \JSL(S,S),\quad \{w_1,\ldots,w_n\}~\mapsto~ \bigvee_{i=1}^n \delta_{w_i} .\] 
Here $\Pow_\f (\Sigma^*)$ is the free idempotent semiring on $\Sigma$, with composition given by concatenation of languages and join given by union. Thus, $\ts{A}$ is the semi\-ring carried by all morphisms $\bigvee_{i=1}^n \delta_{w_i}$ for $w_1,\ldots,w_n\in \Sigma^*$, with join given as above and multiplication $\bigvee_{j} \delta_{v_j} \bullet \bigvee_{i} \delta_{w_i} = \bigvee_{i,j} \delta_{v_jw_i}$. We view $\ts{A}$ as a $\JSL$-automaton with initial state $\id_S = \delta_\epsilon$, final states all $\bigvee_{i}\delta_{w_i}$ such that some $w_i$ is accepted by $A$, and transitions $\bigvee_{i=1}^n \delta_{w_i} \xra{~a~} \bigvee_{i=1}^n \delta_{w_ia}$ for $w_1,\ldots,w_n\in \Sigma^*$ and $a\in \Sigma$. This $\JSL$-automaton is reachable and accepts the same language as $A$. It has the following dual characterization:

\begin{notation}
Given a simple $\JSL$-automaton $A=(S,\delta,i,f)$, the subautomaton
of $\Fin{L}$ obtained by closing $S$ (viewed as a set of languages)
under right derivatives is called the \emph{right-derivative closure}
of $A$ and denoted $\rqc{A}$.
\end{notation}

\begin{proposition}\label{prop:tsdual}
Let $A$ be a reachable $\JSL$-dfa. Then the transition semiring of $A$, viewed as a $\JSL$-dfa, is dual to the right-derivative closure of $A^\op$:
\[ [\ts{A}]^\op \cong \rqc{A^\op} .\]
\end{proposition}
Note that both $[\ts{A}]^\op$ and $\rqc{A^\op}$ are simple, hence subautomata of $\Fin{L}$. Thus, the isomorphism just expresses that their states accept the same languages.


\section{Boolean Representations and Subatomic NFAs}
Based upon the duality results of the previous section, we will now introduce our algebraic approach to nondeterministic state minimality. It rests on the concept of a representation of a monoid on a finite semilattice.

\begin{defn}[Boolean representation]
  Let $M$ be a monoid.  
\begin{enumerate} \item 
A \emph{boolean representation} of $M$ is given by a finite semilattice $S$ together with a monoid morphism $\rho \colon M \to \JSL(S, S)$. The \emph{degree} of $\rho$ is \[\deg(\rho) := |J(S)|.\]
\item Given boolean representations $\rho_i \colon M \to \JSL(S_i, S_i)$, $i=1,2$, an \emph{equivariant map} $f \colon \rho_1 \to \rho_2$ is a $\JSL$-morphism $f \colon S_1 \to S_2$ such that
    \[
      \text{$f (\rho_1(m)(s)) = \rho_2(m)(f(s))$ for all $m \in M$ and $s \in S_1$.}  
    \]
If $f$ is injective, we say that the representation $\rho_2$ \emph{extends} $\rho_1$.
\end{enumerate}
%
\end{defn}

\begin{rem}
  \begin{enumerate}
\item The above representations are called \emph{boolean} because semilattices are precisely semimodules over the boolean semiring $2=\{0,1\}$ with $1+1=1$. For more on representations over general commutative semirings, see~\cite{irs11}.
    \item The category of boolean representations of $M$ coincides with the functor category $\JSL_\f^M$, viewing $M$ as a one object category.
  \end{enumerate}
\end{rem}

\begin{defn}[Canonical representation] For every regular language $L$, the \emph{canonical boolean representation} of the syntactic monoid $\Syn{L}$ is given by
  \[
    \kappa_L\colon
    \Syn{L} \to \JSL(\jslLQ{L}, \jslLQ{L}),\quad [w]_L \mapsto \lambda K.w^{-1}K.
  \]
It induces the \emph{canonical boolean presentation} of the free monoid $\Sigma^*$ given by
 \[
\kappa_L \circ \mu_L\colon \Sigma^*\to \JSL(\jslLQ{L}, \jslLQ{L}),\quad w \mapsto \lambda K.w^{-1}K,
\] 
where $\mu_L\colon \Sigma^*\epito \Syn{L}$ is the syntactic morphism.
\end{defn}
The representation $\kappa_L\circ \mu_L$ amounts to constructing the transition semiring of the minimal $\JSL$-automaton $\jslLQ{L}$, i.e.\ the \emph{syntactic semiring}~\cite{Polak2001} of $L$.
\begin{expl}
   We describe the canonical boolean representation $\kappa_{L_n}$ for the language $L_n := (0 + 1)^* 1 (0 + 1)^n$, $n\in\Nat$. Let $S:=2^{n + 1}_\bot$ be the semilattice of binary words of length $n + 1$, ordered pointwise, with an additional bottom element $\bot$. Then $\jslLQ{L_n}$ is isomorphic to $S$, as witnessed by the isomorphism
\[f \colon S \xra{\cong} \LQ{L_n},\quad f(\bot)=\emptyset, \quad f(w) = w^{-1} L_n.\] Thus, $\kappa_{L_n}$ is isomorphic to the representation $\rho\colon \Syn{L_n} \to \JSL(S,S)$ where:
  \begin{enumerate}
    \item $\rho([0]_{L_n})\colon S \to S$ performs a left-shift (distinct from left-rotate);
    \item $\rho([1]_{L_n})\colon S \to S$ performs a left-shift and sets the last bit as $1$.
  \end{enumerate}
  Finally, $\deg(\kappa_{L_n}) = \deg(\rho) = 1 + |J(2^{n + 1})| = n + 2$ is the number of states of the usual minimal nfa for $L$.
%
\end{expl}

\begin{expl}
  \label{ex:lang_with_m3}
We describe the canonical boolean presentation $\kappa_L$ for the
language $L = a_1(a_2 + a_3) + a_2(a_1 + a_3) + a_3(a_1 + a_2)$ over
$\Sigma=\{a_1,a_2,a_3\}$. Consider the $\cup$-semilattice $M_3=\{
\emptyset, \{a_1,a_2\}, \{a_1,a_3\}, \{a_2,a_3\}, \Sigma \}$. Then
$\jslLQ{L}$ is isomorphic to the product semilattice $2\times
M_3\times 2$ via the map
\[f\colon \jslLQ{L} \xra{\cong} 2 \times M_3 \times 2,\quad f(X) = (X \cap \Sigma^2, X \cap \Sigma, X \cap \{ \epsilon \}).\] 
Note that the first and third component is either
$\emptyset$ or one other set, i.e.\ it may be identified with the
elements of $2$. For $i=1,2,3$ we define the following semilattice morphisms:
\begin{align*}
&\alpha_i\colon 2\to M_3, &&  \alpha_i(1)=\Sigma\setminus \{a_i\}; \\
&\beta_i\colon M_3\to 2, && \beta_i(S) = 1 \iff a_i\in S; \\
&\gamma\colon 2\to 2 &&  \gamma(1) = 0; \\
&\delta\colon M_3\times 2\times 2\to 2\times M_3 \times 2,&& \delta(x,y,z)=(z,x,y).
\end{align*}
Then $\kappa_L$ is isomorphic to $\rho\colon \Syn{L}\to \JSL(2\times M_3\times 2,\, 2\times M_3\times 2)$ where
  \[
    \rho([a_i]_L) = (\,2\times M_3\times 2 \xra{\alpha_i\times \beta_i\times \gamma} M_3\times 2\times 2 \xra{\delta} 2\times M_3\times 2\,).
  \]
Thus, $\deg(\kappa_L) = \deg(\rho) = 1 + 3 + 1 = 5$. An analogous description of $\kappa_L$ exists for any language $L$ where each word has the same length. 
%
\end{expl}
The next theorem links minimal nfas and representations. 

\begin{defn} The \emph{nondeterministic state complexity} $\ns{L}$
 of a regular language $L$ is the least number of states of any nfa accepting $L$.
\end{defn}

\begin{theorem}
  \label{thm:ns_char}
For every regular language $L$, the  nondeterministic state complexity
  $\ns{L}$ is the least degree of any boolean representation extending the canonical representation $\kappa_L \o \mu_L\colon \Sigma^*\to \JSL(\LQ{L},\LQ{L})$.
\end{theorem}


\begin{proof}[Sketch]
\begin{enumerate}
\item Given a $k$-state nfa $N=(Q,\delta,I,F)$ accepting $L$, consider the subsemilattice $\jslLangs{N}=\simple{\Pow(N)}$ of $\Pow(\Sigma^*)$ on all languages accepted by subsets of $Q$. The embedding $\jslLQ{L}\monoto \jslLangs{N}$ yields an extension of $\kappa_L\o \mu_L$. Since the semilattice $\jslLangs{N}$ is generated by the languages accepted by single states of $N$, this extension has degree at most $k$.
\item Conversely, let $\rho\colon \Sigma^*\to \JSL(S,S)$ be a boolean representation of degree $k$ extending $\kappa_L\o \mu_L$, witnessed by an injective equivariant map $h\colon \LQ{L}\monoto S$. One can equip $S$ with a $\JSL$-dfa structure making $h$ an automata morphism. Since morphisms preserve accepted languages, it follows that $S$ accepts $L$. Then the nfa of join-irreducibles of $S$, see \autoref{rem:jsldfa_vs_nfa}, is a $k$-state nfa accepting $L$.\qed
\end{enumerate}
\end{proof}
As an application, let us return to the dependency relation $\rDR{L}$ introduced in \autoref{rem:conway}(2). Recall that a \emph{biclique} of a relation $R\seq X\times Y$ (viewed as a bipartite graph) is a subset of the form $X'\times Y'\seq R$, where $X'\seq X$ and $Y'\seq Y$. A \emph{biclique cover} of $R$ is a set $\C$ of bicliques with $R=\bigcup \C$. The \emph{bipartite dimension} $\dim{R}$ is the least cardinality of any biclique cover of $R$.

\begin{theorem}[Gruber-Holzer \cite{NDLowBoundsHard2006}]\label{thm:dim_vs_ns}
  For every regular language $L$, we have \[\dim{\rDR{L}} \leq \ns{L}.\]
\end{theorem}
We give a new algebraic proof of this result based on boolean representations.

\begin{proof}
\begin{enumerate}
\item The task of computing biclique covers is well-known to be equivalent to the \emph{set basis} problem. Given a family $C\seq \Pow(Y)$ of subsets of a finite set $Y$, a {set basis} for $C$ is a family $B\seq \Pow(Y)$ such that each element of $C$ can be expressed as a union of elements of $B$. A relation $R\seq X\times Y$ has a biclique cover of size $k$ iff the family $C_R=\{R[x] : x\in X\}\seq \Pow(Y)$ of neighborhoods of nodes in $X$ has a set basis of size $k$.
\item Given an instance $C\seq \Pow(Y)$ of the set basis problem, consider the $\cup$-subsemilattice $\langle C\rangle \seq \Pow(Y)$ generated by $C$, i.e.~the semilattice of all unions of sets in $C$. We claim that $C$ has a set basis of size at most $k$ iff there exists an extension of $\langle C\rangle$ of degree at most $k$, i.e.~a monomorphism $\langle C\rangle\monoto S$ into some finite semilattice $S$ with $\under{J(S)}\leq k$. 

\smallskip\noindent For the ``only if'' direction, suppose that $B\seq \Pow(Y)$ is a set basis of $C$ of size at most $k$. The the embedding $\langle C\rangle \monoto \langle B \rangle$  gives an extension of $\langle C\rangle$ with the desired property: since the semilattice $\langle B\rangle$ has a set of generators with at most $k$ elements, it has at most $k$ join-irreducibles.

\smallskip\noindent For the ``if'' direction, suppose that $m\colon \langle C\rangle\monoto S$ with $\under{J(S)}\leq k$ is given. Since the free semilattice $\Pow(Y)$ is an injective object of $\JSL$~\cite[Corollary 2.9]{hk71}, there exists a morphism $f\colon S\to \Pow(Y)$ extending the embedding $\langle C\rangle \monoto \Pow(Y)$. Consider the image $S'\seq \Pow(Y)$ of $f$, leading to the commutative diagram below:
\[
\xymatrix@R-1em{
\langle C\rangle \ar@{>->}[dr]_\seq \ar@{>->}[r]^m & S \ar[d]^<<<<f \ar@{->>}[r]^e & S' \ar@{>->}[dl]^\seq \\
& \Pow(Y) & 
}
\]   
We thus have $\langle C\rangle \seq S'\seq \Pow(Y)$. Every set of generators of the semilattice $S'$ is a basis of $C$. Since the morphism $e$ is surjective, we have $\under{J(S')}\leq \under{J(S)}\leq k$, i.e.~$S'$ has a set of generators with at most $k$ elements.
\item Let $C_{\rDR{L}}\seq \Pow(\LW{\rev{L}})$ be the instance of the set basis problem corresponding to the dependency relation $\rDR{L}\seq \LW{L} \times \LW{\rev{L}}$. Note that $\langle C_{\rDR{L}}\rangle$ consists of all $\rDR{L}[X]$ for $X\seq \LW{L}$. Thus, \autoref{thm:dep_thm}(1) shows that $\langle C_{\rDR{L}}\rangle \cong \jslLQ{L}$. In particular, every extension of the canonical boolean representation of $\Sigma^*$ yields an extension of the semilattice $\langle C_{\rDR{L}} \rangle$ of the same degree. Therefore, by part (1) and (2) and \autoref{thm:ns_char}, we have $\dim{\rDR{L}} \leq \ns{L}$, as required.
\end{enumerate}
\end{proof}
\autoref{thm:ns_char}  motivates the following definition, which can be considered the key concept of our paper:
\begin{defn}The \emph{nondeterministic syntactic complexity} 
  $\nsyn{L}$ of a regular language $L$ is the least degree of any boolean representation of $\Syn{L}$ extending the canonical boolean representation $\kappa_L\colon \Syn{L}\to \JSL(\LQ{L},\LQ{L})$.
\end{defn}
Just like the degrees of boolean representations of $\Sigma^*$ determine the state complexity of nfas, we will provide an automata-theoretic characterization of $\nsyn{L}$ in terms of \emph{subatomic} nfas in \autoref{thm:nmu_eq_subatomic} below.


\begin{defn}
An nfa accepting the language $L$ is called 
\begin{enumerate}
\item \emph{atomic} if each state accepts a language from $\jslBLQ{L}$, and
\item \emph{subatomic} if each state accepts a language from $\jslBLRQ{L}$.
\end{enumerate}
\end{defn}
The notion of an atomic nfa goes back to Brzozowski and Tamm~\cite{TheoryOfAtomataBrzTamm2014}, as does the following characterization. 

\begin{notation} For any nfa $N$, let $\rsc{N}$ denote the dfa obtained via the \emph{reachable subset construction}, i.e.~the dfa-reachable part of $\Pow(N)$.
\end{notation}

\begin{theorem}
  \label{thm:atomic_nfa_char}
  An nfa $N$ is atomic iff $\rsc{\rev{N}}$ is a minimal dfa.
\end{theorem}
We present a new conceptual proof, interpreting this theorem as an instance of the self-duality of $\JSL$-dfas.
\begin{proof}[Sketch]
Let $L$ be the language accepted by $N$. We establish the theorem by showing each of the following statements to be equivalent to the next one:
\begin{enumerate}
\item $N$ is atomic.
\item There exists a $\JSL$-automata morphism from $\Pow(N)$ to $\jslBLQ{L}$. 
\item There exists a $\JSL$-automata morphism from $\Pow(\minDfa{\rev{L}})$ to ${\Pow(\rev{N})}$.
\item There exists a dfa morphism from $\minDfa{\rev{L}}$ to ${\Pow(\rev{N})}$.
\item There exists a dfa morphism from $\minDfa{\rev{L}}$ to $\rsc{\rev{N}}$.
\item $\rsc{\rev{N}}$ is a minimal dfa.
\end{enumerate}
\newcommand{\lemmaautorefname}{Lemmas}
The key step is (2)$\Lra$(3), which follows via duality from
\autoref{lem:nfarev} and~\ref{lem:astarprops}, and \autoref{prop:blqdual}. All remaining equivalences follow from the definitions.\qed
\end{proof}
The next theorem gives an analogous characterization of subatomic nfas. Again, the proof is based on duality. 

\begin{theorem}
  \label{thm:subatomic_nfa_char}
  An nfa $N$ accepting the language $L$ is subatomic iff the transition monoid of $\rsc{\rev{N}}$ is isomorphic to the syntactic monoid $\Syn{\rev{L}}$.
\end{theorem}

\begin{proof}[Sketch]
Each of the following statements is equivalent to the next one:
\begin{enumerate}
\item $N$ is subatomic.
\item There exists a $\JSL$-dfa morphism from $\Pow(N)$ to $\jslBLRQ{L}$.
\item There exists a $\JSL$-dfa morphism from $\rqc{\simple{P(N)}}$ to $\jslBLRQ{L}$.
\item There exists a $\JSL$-dfa morphism from $\Pow(\Syn{\rev{L}})$ to $\ts{\reach{\Pow(\rev{N})}}$.
\item There exists a dfa morphism from $\Syn{\rev{L}}$ to $\ts{\reach{\Pow(\rev{N})}}$.
\item There exists a dfa morphism from $\Syn{\rev{L}}$ to $\tm{\rsc{\rev{N}}}$.
\item The monoids $\Syn{\rev{L}}$ and $\tm{\rsc{\rev{N}}}$ are isomorphic.
\end{enumerate}
The equivalence (3)$\Lra$(4) follows via duality from
\autoref{lem:nfarev}, \autoref{prop:blrqdual} and~\autoref{prop:tsdual}. All remaining equivalences follow from the definitions.\qed
\end{proof}
We are prepared to state the main result of our paper, an automata-theoretic characterization of the nondeterministic syntactic complexity:
\begin{theorem}\label{thm:nmu_eq_subatomic} For every regular language $L$, the nondeterministic syntactic complexity $\nsyn{L}$ is the least number of states of any subatomic nfa accepting $L$.
\end{theorem}

\begin{proof}[Sketch]
\begin{enumerate}
\item Let $N$ be a $k$-state subatomic nfa accepting the language $L$. As in the proof of \autoref{thm:ns_char}, we consider the semilattice $\jslLangs{N}=\simple{\Pow(N)}$. Then
\[ \rho\colon \Syn{L}\to \JSL(\jslLangs{N},\jslLangs{N}),\quad [w]_L \mapsto \lambda K.w^{-1}K,\]
is a representation of $\Syn{L}$ of degree at most $k$ extending $\kappa_L$.
\item Conversely, let $\rho\colon \Syn{L}\to \JSL(S,S)$ be a boolean representation extending $\kappa_L$, and let $h\colon \jslLQ{Q}\monoto S$ be the embedding. As in the proof of \autoref{thm:ns_char}, we can equip $S$ with the structure of a $\JSL$-dfa making $h$ an automata morphism. Its nfa of join-irreducibles, see \autoref{rem:jsldfa_vs_nfa}, is a subatomic nfa accepting $L$ with $\deg(\rho)$ states.\qed
\end{enumerate}
\end{proof}
We conclude this section with the observation that the state complexity of unrestricted nfas, subatomic nfas and atomic nfas generally differs:

\begin{expl}[Subatomic more succinct than atomic]
Consider the language $L$ accepted by the nfa $N$ shown below, along with the minimal dfas for $L$ and $\rev{L}$. Each automaton has exactly one initial state, namely $0$.
  {\tiny
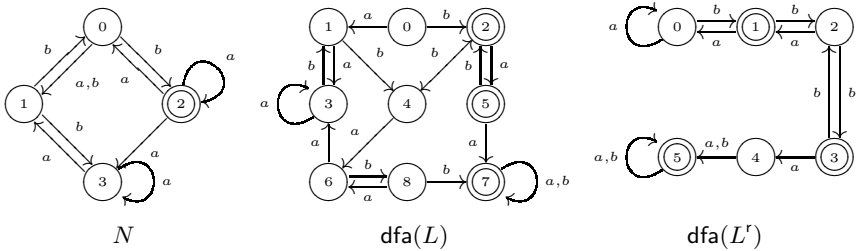
  \[
    \begin{tabular}{lll}
      \xymatrix@=14pt {
        & *++[o][F-]{0} \ar@<2pt>[dl]^{a,b} \ar@<2pt>[dr]^b \\
        *++[o][F-]{1} \ar@<2pt>[dr]^b \ar@<2pt>[ur]^b & & *++[o][F=]{2} \ar@<2pt>[ul]^a \ar@(u,r)^a \ar[dl]^a \\
        & *++[o][F-]{3} \ar@(ur,dr)^a \ar@<2pt>[ul]^a
      }
      &
      \xymatrix@=14pt {
        *++[o][F-]{1} \ar[dr]^b \ar@<2pt>[d]^a & *++[o][F-]{0} \ar[l]_a \ar[r]^b  & *++[o][F=]{2} \ar@{<->}[dl]_b \ar@<2pt>[d]^a
        \\
        *++[o][F-]{3} \ar@(dl,ul)^a \ar@<2pt>[u]^b & *++[o][F-]{4} \ar[dl]_a & *++[o][F=]{5} \ar@<2pt>[u]^b \ar[d]_a
        \\
        *++[o][F-]{6} \ar[u]^a \ar@<2pt>[r]^b & *++[o][F-]{8} \ar@<2pt>[l]^a \ar[r]^b & *++[o][F=]{7} \ar@(ur,dr)^{a,b}
      }
      &
      \xymatrix@=14pt {
        *++[o][F-]{0} \ar@(dl,ul)^a \ar@<2pt>[r]^b & *++[o][F=]{1} \ar@<2pt>[l]^a \ar@<2pt>[r]^b & *++[o][F-]{2} \ar@<2pt>[l]^a \ar@<2pt>[dd]^b
        \\ \\
        *++[o][F=]{5} \ar@(dl,ul)^{a,b} & *++[o][F-]{4} \ar[l]_-{a,b} & *++[o][F=]{3} \ar@<2pt>[uu]^b \ar[l]^a
      }
      \\
      \\
      \multicolumn{1}{c}{\small $N$} &
      \multicolumn{1}{c}{\small $\minDfa{L}$} &
      \multicolumn{1}{c}{\small $\minDfa{\rev{L}}$}
    \end{tabular}
  \]
  }\noindent
 Brzozowski and Tamm~\cite{TheoryOfAtomataBrzTamm2014} showed that there is no atomic nfa with four states accepting $L$. However, $N$ is subatomic: one can verify that the transition monoids of $\minDfa{\rev{L}}$ and $\rsc{\rev{N}}$ both have $22$ elements. Since the former is the syntactic monoid of $\rev{L}$, they are isomorphic, and so \autoref{thm:subatomic_nfa_char} applies.
%
%
%
\end{expl}

\begin{expl}[Subatomic less succinct than general nfas]
  There is a regular language for which no state-minimal nfa is subatomic: 
  \[
    L \;:=\; \{\, a^n \;:\; n\in \Nat,\, n \neq 5 \,\} \seq \{a\}^*.
  \]
  It is accepted by the following nfa:
{\tiny
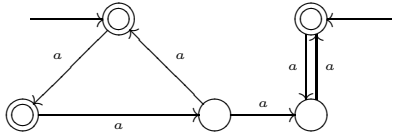
  \[
    \xymatrix {
    \ar[r] & *++[o][F=]{} \ar[dl]_a & & *++[o][F=]{} \ar@<-2pt>[d]_a & \ar[l] \\
    *++[o][F=]{} \ar[rr]_a & & *++[o][F-]{} \ar[ul]_a \ar[r]^a & *++[o][F-]{} \ar@<-2pt>[u]_a  
    }
 \]
}\noindent
An exhaustive search shows that no subatomic nfa with five states accepts $L$. In fact, $L$ is the unique (!) unary language with $\ns{L}\leq 5$ and $\ns{L}<\nsyn{L}$. Moreover, the above nfa and its reverse are the only state-minimal nfas for $L$.
\end{expl}

\section{Applications}\label{S:examples}

While subatomic nfas are generally less succinct then unrestricted ones, all structural results concerning nondeterministic state complexity we have encountered in the literature are actually about nondeterministic syntactic complexity: they implicitly identify classes of languages where the two measures coincide. In the present section, we illustrate this in a few selected applications.

\subsection{Unary languages}

For unary languages $L\seq\{a\}^*$, two-sided derivatives are left derivatives. Thus, a unary nfa is atomic iff it is subatomic.


\begin{expl}[Cyclic unary languages]
  \label{ex:cyclic_unary}
  A unary language $L$ is \emph{cyclic} if its minimal dfa is a cycle \cite{Gramlich03probabilisticand}.  We claim that $\ns{L} = \nsyn{L}$. To see this, let $d:=\under{\LW{L}}$ be the \emph{period} (i.e. number of states) of the minimal dfa. By Fact 1 of \cite{Gramlich03probabilisticand} (originally from \cite{MinNfaUnary91}) every state-minimal nfa $N$ accepting $L$ is a disjoint union of cyclic dfas whose periods divide $d$.\footnote{In \cite{Gramlich03probabilisticand} nfas are restricted to have a single initial state and so are distinguished from unions of dfas; the latter are valid nfas from our perspective.} Then $|\rsc{\rev{N}}| = d$: we have $|\rsc{\rev{N}}| \geq d$ since $\rsc{\rev{N}}$ is a dfa accepting $L=\rev{L}$ and $d$ is the size of the minimal dfa for $L$, and
 $|\rsc{\rev{N}}| \leq d$  because  after $d$ steps, each cycle will be back in its initial state. Thus $N$ is atomic by Theorem \ref{thm:atomic_nfa_char} and hence subatomic. 
\end{expl}
We deduce the following result for (not necessarily unary) regular languages:

\begin{theorem}  \label{thm:cyclic_syn_mon}
If $\Syn{L}$ is a cyclic group, then
$\ns{L}=\nsyn{L}$.
\end{theorem}

\begin{proof}[Sketch]
Suppose that $\Syn{L}=\tm{\minDfa{L}}$ is cyclic. Then there exists $w_0\in\Sigma^*$ such that the map $\lambda X. w_0^{-1} X\colon \LW{L} \to \LW{L}$ generates $\tm{\minDfa{L}}$. Fix an alphabet $\Sigma_0 = \{ a_0 \}$ disjoint from $\Sigma$ and consider the unary language
\[
  L_0 := \{\, a_0^n \;:\; n \in \Nat,\, w_0^n \in L \,\} \subseteq \Sigma_0^*.
\]
 Let $g: \Sigma_0^* \to \Sigma^*$ be the monoid morphism where $g(a_0) := w_0$. Then we have the $\JSL$-isomorphism
\[f \colon\jslLQ{L_0} \xra{\cong} \jslLQ{L},\quad f(X^{-1} L_0) := [g[X]]^{-1} L.\] 
For each $a \in \Sigma$ choose $n_a\in \Nat$ such that $a^{-1} K = (w_0^{n_a})^{-1} K$ for all $K\in \LW{L}$. The respective transition endomorphisms of the $\JSL$-automata $\jslLQ{L_0}$ and $\jslLQ{L}$ determine each other in the sense that the following diagrams commute:
\[
\xymatrix@R-1em{
\jslLQ{L_0} \ar[r]^f_\cong \ar[d]_{a_0^{-1}(\dash)} & \jslLQ{L} \ar[d]^{ w_0^{-1}(\dash)  }  \\
\jslLQ{L_0} \ar[r]_f^\cong & \jslLQ{L}
}
\qquad
\xymatrix@R-1em{
\jslLQ{L_0} \ar[r]^f_\cong \ar[d]_{(a_0^{n_a})^{-1}(\dash)} & \jslLQ{L} \ar[d]^{ a^{-1}(\dash)  }  \\
\jslLQ{L_0} \ar[r]_f^\cong & \jslLQ{L}
}
\]
Then $\ns{L} = \ns{L_0}$ by \autoref{thm:ns_char} and $\nsyn{L} = \nsyn{L_0}$ by \autoref{thm:nmu_eq_subatomic}. Moreover, by \autoref{ex:cyclic_unary} we know that $\ns{L_0}=\nsyn{L_0}$, so the claim follows.
\end{proof}

\begin{expl}[$\nsyn{L}$ no larger than Chrobak normal form]
  \label{ex:chrobak}
A unary nfa is in \emph{Chrobak normal form}~\cite{chrobak1986finite,ChrobakRevisited2011} if it has a single initial state and at most one state with multiple successors, all of which lie in disjoint cycles.
We claim that for any nfa $N$ in Chrobak normal form accepting the language $L$, we have 
\[ \nsyn{L} \leq |N|,\]
where $\under{N}$ denotes the number of states of $N$. To see this, observe that each state of $N$ up to and including the unique choice state accepts some left derivative of $L$. The successors of the choice state  collectively accept a derivative $u^{-1} L$; this language is cyclic because it is a finite union of cyclic languages. Therefore, by Example \ref{ex:cyclic_unary} we may replace the cycles by an atomic nfa accepting $u^{-1} L$, without increasing the number of states. The resulting nfa is atomic.

Since every unary nfa on $n$ states can be transformed into an nfa in Chrobak normal form with $O(n^2)$ states~\cite[Lemma 4.3]{chrobak1986finite}, we get:
\end{expl}

\begin{corollary}
If $L$ is a unary regular language, then $\nsyn{L} = O(\ns{L}^2)$.
\end{corollary}

%

\subsection{Languages with a canonical state-minimal nfa}
There are several natural classes of regular languages for which \emph{canonical} state-minimal nondeterministic acceptors have been identified. We show that these acceptors are actually subatomic. In our arguments, we frequently consider the \emph{length} of a finite semilattice $S$, i.e.~the maximum length $n$ of any ascending chain $s_0<s_1<\ldots <s_n$ in $S$. Note that since every element is uniquely determined by the set of join-irreducibles below it, the length of $S$ is at most $\under{J(S)}$.

\begin{expl}[Bideterministic and biseparable languages]
  \label{ex:bisep}
\begin{enumerate}
 \item  A language is called \emph{bideterministic} if it is accepted by a dfa whose reverse is also a dfa. In this case, the minimal dfa is a minimal nfa \cite{ReversibleAutomata92,TAMM2004135}. Bideterministic languages have been studied in the context of automata learning~\cite{angluin82} and coding theory, where they are known as \emph{rectangular codes}~\cite{TrellisMFCCodes96, BlockCodesFiniteAut2003}. We show that for every bideterministic language $L$,
\[
\ns{L}=\nsyn{L}=\under{\LW{L}}.
\] 
To this end, we first note that by~\cite[Theorem 3.1]{BlockCodesFiniteAut2003} a language $L\seq \Sigma^*$ is bideterministic iff the left derivatives of $L$ are pairwise disjoint. This implies that $\jslLQ{L}$ is a boolean algebra with atoms $\LW{L}$. Since the length of a boolean algebra equals the number of atoms (= join-irreducibles), we conclude that for every finite semilattice extension $\LQ{L} \monoto S$, the semilattice $S$ has length at least $\under{\LW{L}}$. Thus, $|\LW{L}| \leq |J(S)|$, so any representation $\rho$ extending $\kappa_L$ or $\kappa_L\o \mu_L$ satisfies $|\LW{L}| \leq \deg(\rho)$. Hence, $\ns{L} = \nsyn{L} = |\LW{L}|$ by \autoref{thm:ns_char} and \ref{thm:nmu_eq_subatomic}. In particular, the minimal dfa of $L$ is a minimal nfa.
\item A language $L$ is \emph{biseparable} if $\jslLQ{L}$ is a boolean algebra \cite{MinNfaBiRFSA2009}.\footnote{Actually \cite{MinNfaBiRFSA2009} defines biseparability as a property of nfas, and characterizes biseparable nfas as those accepting a language $L$ for which no $\cup$-irreducible left derivative is contained in the union of other $\cup$-irreducible left derivatives. This is equivalent to the lattice $\jslLQ{L}$ being boolean, i.e.~ to $L$ being `biseparable' in our sense.}
For every biseparable language $L$, the \emph{canonical residual automaton}~\cite{ResidFSA2001}, i.e.~the nfa $N_L$ of join-irreducibles of the minimal $\JSL$-dfa $\jslLQ{L}$, is a state-minimal nfa; it is subatomic because every state of $N_L$ accepts a derivative of $L$. This follows exactly as in (1): our argument only used that $\jslLQ{L}$ is a boolean algebra.
\end{enumerate}
\end{expl}

\begin{expl}[Maximal reachability] A folklore result asserts that if $N$ is an nfa whose accepted language $L$ satisfies  $|\LW{L}| = 2^{|N|}$, then $N$ is state-minimal. Since $\LD{L}$ forms the set of states of the minimal dfa for $L$ and $\rsc{N}$ accepts $L$, we have $\rsc{N} = \Pow(N)$. It follows the $\JSL$-dfa $\Pow(N)$ is reachable and simple, hence isomorphic to the minimal $\JSL$-dfa $\jslLQ{L}$. This proves that $\jslLQ{L}$ is a boolean algebra, i.e.\ $L$ is a biseparable language. We conclude from \autoref{ex:bisep}(2) that $\ns{L}=\nsyn{L}=\under{N}$ and $N_L$ is a subatomic minimal nfa.
\end{expl}

\begin{expl}[BiRFSA and topological languages]
  \label{ex:biresid_topological}
  So far $\jslLQ{L}$ has been a boolean algebra. But the argument in Example \ref{ex:bisep} {also} applies when $\jslLQ{L}$ is a  distributive lattice, noting that the length of a finite distributive lattice  is equal to the number of its join-irreducibles~\cite[Corollary 2.14]{GratzerGeneralLattice1998}. Languages with this property are called \emph{topological}~\cite{CtsNfa2014}. It thus follows as in \autoref{ex:bisep}(2) that for any topological language $L$, the {canonical residual automaton} $N_L$ is subatomic and a state-minimal nfa. Thus, $\ns{L} = \nsyn{L} = |J(\jslLQ{L})|$.

  There is another class of languages where $N_L$ is known to be a state-minimal nfa, the \emph{biRFSA} languages \cite{MinNfaBiRFSA2009}. A language $L$ is called {biRFSA} if $N_L$ is isomorphic to $\rev{(N_{\rev{L}})}$. Surprisingly, these languages are exactly the topological ones:
  \begin{enumerate}
    \item
    \emph{Suppose that $L$ is topological}. Recall that $N_L$ is the nfa of join-irreducibles of the minimal $\JSL$-dfa. Thus, it has states $J(\jslLQ{L})$ and transitions given by
$ X \xto{a} Y$ iff $Y \subseteq a^{-1} X$ for $a\in \Sigma$.  Moreover, a join-irreducible $j$ is initial iff $j \subseteq L$ and final iff $\epsilon \in j$. Since the lattice $\jslLQ{L}$ is distributive, we have a canonical bijection between its join- and meet-irreducibles:
\[ \tau\colon  J(\jslLQ{L}) \xra{\cong} M(\jslLQ{L}),\quad \tau(j) = \bigcup \{ X \in \jslLQ{L} : j \nsubseteq X \}.\] 
Let $\theta$ be the unique map making the following diagram commute, where $\dr_L$ is the restriction of the isomorphism of \autoref{prop:lqdual}:
\[ 
\xymatrix@C-2em@R-1em{
& J(\jslLQ{L}) \ar[dr]^{\tau}_\cong  \ar[dl]_\theta^\cong & \\
J(\jslLQ{\rev{L}}) \ar[rr]_{\dr_L}^\cong & & M(\jslLQ{L})
} \]
One can show $\theta$ to be an nfa isomorphism from $N_L$ to $\rev{({N_{\rev{L}}})}$.  Thus, $L$ is biRFSA.

    \item \emph{Suppose that $L$ is biRFSA.} Then we have a surjective $\JSL$-morphism
\[ [\Pow(J(\jslLQ{{L}}))]^\op \cong \Pow(J(\jslLQ{\rev{L}})) \xra{e_{\rev{L}}} \jslLQ{\rev{L}} \cong [\jslLQ{L}]^\op,\]
where the first isomorphism follows from $N_L \cong \rev{({N_{\rev{L}}})}$ and \autoref{lem:nfarev}, the second isomorphism is given by \autoref{prop:lqdual}, and $e_{\rev{L}}$ sends $X\seq J(\jslLQ{\rev{L}})$ to $\bigcup X$. The dual of this morphism is the injective $\JSL$-morphism
\[ m_L\colon \jslLQ{L}\monoto \Pow(J(\jslLQ{{L}}))\]
sending $K\in  \jslLQ{L}$ to the set of all $j\in J(\jslLQ{L})$ with $j\seq K$. Note that $e_L\o m_L = \id_{\jslLQ{Q}}$, showing that $\jslLQ{L}$ is a retract of $\Pow(J(\jslLQ{L}))$. Since $\JSL$-retracts of finite distributive lattices are distributive, see e.g.~\cite[Lemma 2.2.3.15]{myers2020representing}, it follows that $\jslLQ{L}$ is distributive. Thus, $L$ is topological.
%
  \end{enumerate}
\end{expl}

\begin{expl}[Extremal languages]
  Call a language \emph{extremal} if $\jslLQ{L}$ has length $|J(\jslLQ{L})|$ i.e.\ we have an \emph{extremal lattice} in the sense of Markowsky \cite{ExtremalLat92}. Again, the argument of Example \ref{ex:bisep} applies and we get $\ns{L} = \nsyn{L} = |J(\jslLQ{L})|$. Topological languages are extremal since every distributive lattice is an extremal lattice, although extremal languages need not be topological. Both classes are naturally characterized in terms of the reduced dependency relation:
\begin{enumerate}
  \item 
  $L$ is topological iff $\rRDR{L}$ is essentially an order relation $\leq_P \,\subseteq P\times P$ of a finite poset~\cite[Example 2.2.12]{myers2020nondeterministic}.
  \item
  $L$ is extremal iff $\rRDR{L}$ is \emph{upper unitriangularizable} \cite[Theorem 11]{ExtremalLat92}.
\end{enumerate}
The latter means the adjacency matrix of the bipartite graph $\rRDR{L}$ can be put in upper triangular form with ones along the diagonal, by permuting rows and columns. An order relation is upper unitriangularizable because it may be extended to a linear order.
\end{expl}

\section{Conclusion and Future Work}\label{S:conclusion}

Motivated by the duality theory of deterministic finite automata over semilattices,
we introduced a natural class of nondeterministic finite automata called \emph{subatomic nfas} and studied their state complexity in terms of boolean representations of syntactic monoids. Furthermore, we demonstrated that a large body of previous work on state minimization of general nfas actually constructs minimal subatomic ones. There are several directions for future work.

As illustrated by \autoref{thm:dim_vs_ns}, the dependency relation $\rDR{L}$ forms a useful tool for proving lower bounds on nfas. It is also a key element of the Kameda-Weiner algorithm~\cite{KamedaWeiner1970,tamm16} for minimizing nfas, which rests on computing biclique covers of $\rDR{L}$. We aim to give an algebraic interpretation of dependency relations based on the representation of finite semilattices by contexts~\cite{ContextJipsen2012}, which can be augmented to a categorical equivalence between $\JSL_\f$ and a suitable category of bipartite graphs~\cite{myers2020representing}. Under this equivalence, $\JSL$-dfas correspond to \emph{dependency automata}; in particular, the minimal $\JSL$-dfa $\LQ{L}$ corresponds to a dependency automaton whose underlying bipartite graph is precisely the dependency relation $\rDR{L}$. We expect that this observation can lead to a fresh algebraic perspective on the Kameda-Weiner algorithm, as well as a generalization of it computing minimal (sub-)atomic nfas.

On a related note, we also intend to investigate the complexity of the minimization problem for (sub-)atomic nfas. While minimizing general nfas is $\PSPACE$-complete, even if the input automaton is a dfa, we conjecture that the additional structure present in (sub-)atomic acceptors will simplify their minimization to an $\NP$-complete task. First evidence in this direction is provided by Geldenhuys, van der Merve, and van Zijl~\cite{NfaSat} whose work implies that minimal atomic nfas can be efficiently computed in practice using SAT solvers.


  



\clearpage
\bibliographystyle{splncs03}
\bibliography{refs,bib-2019,bib-2020}

\begin{thebibliography}{10}
\providecommand{\url}[1]{\texttt{#1}}
\providecommand{\urlprefix}{URL }

\bibitem{CtsNfa2014}
Adámek, J., Myers, R.S., Urbat, H., Milius, S.: On continuous nondeterminism
  and state minimality. In: Proc.~30th Conference on the Mathematical
  Foundations of Programming Semantics (MFPS XXX). vol. 308, pp. 3--23 (2014)

\bibitem{angluin82}
Angluin, D.: Inference of reversible languages. J. {ACM}  29(3),  741--765
  (1982)

\bibitem{am75}
Arbib, M.A., Manes, E.G.: Adjoint machines, state-behavior machines, and
  duality. Journal of Pure and Applied Algebra  6(3),  313--344 (1975)

\bibitem{arbib_manes_1975}
Arbib, M.A., Manes, E.G.: Fuzzy machines in a category. Bulletin of the
  Australian Mathematical Society  13(2),  169–210 (1975)

\bibitem{BACKHOUSE2016824}
Backhouse, R.: Factor theory and the unity of opposites. Journal of Logical and
  Algebraic Methods in Programming  85(5, Part 2),  824--846 (2016)

\bibitem{TheoryOfAtomataBrzTamm2014}
Brzozowski, J., Tamm, H.: Theory of átomata. Theoretical Computer Science
  539,  13--27 (2014)

\bibitem{BrzozowskiDRE1964}
Brzozowski, J.A.: Derivatives of regular expressions. J. ACM  11(4),  481--494
  (Oct 1964)

\bibitem{chrobak1986finite}
Chrobak, M.: Finite automata and unary languages. Theoretical Computer Science
  47,  149--158 (1986)

\bibitem{cm19}
Clemente, L., Mayr, R.: Efficient reduction of nondeterministic automata with
  application to language inclusion testing. {Logical Methods in Computer
  Science}  {Volume 15, Issue 1} (2019)

\bibitem{Conway71}
Conway, J.H.: Regular Algebra and Finite Machines. Printed in GB by William
  Clowes \& Sons Ltd (1971)

\bibitem{ddhr06}
De~Wulf, M., Doyen, L., Henzinger, T.A., Raskin, J.F.: Antichains: A new
  algorithm for checking universality of finite automata. In: Ball, T., Jones,
  R.B. (eds.) Computer Aided Verification. pp. 17--30. Springer (2006)

\bibitem{ResidFSA2001}
Denis, F., Lemay, A., Terlutte, A.: Residual finite state automata. In:
  Ferreira, A., Reichel, H. (eds.) STACS 2001: 18th Annual Symposium on
  Theoretical Aspects of Computer Science Dresden, Germany, February 15--17,
  2001 Proceedings. pp. 144--157. Springer Berlin Heidelberg, Berlin,
  Heidelberg (2001)

\bibitem{ChrobakRevisited2011}
Gawrychowski, P.: Chrobak normal form revisited, with applications. In:
  Bouchou-Markhoff, B., Caron, P., Champarnaud, J.M., Maurel, D. (eds.)
  Implementation and Application of Automata. pp. 142--153. Springer Berlin
  Heidelberg, Berlin, Heidelberg (2011)

\bibitem{NfaSat}
Geldenhuys, J., van~der Merwe, B., van Zijl, L.: Reducing nondeterministic
  finite automata with {SAT} solvers. In: Yli-Jyr{\"a}, A., Kornai, A.,
  Sakarovitch, J., Watson, B. (eds.) Finite-State Methods and Natural Language
  Processing. pp. 81--92. Springer Berlin Heidelberg, Berlin, Heidelberg (2010)

\bibitem{goguen75}
Goguen, J.A.: Discrete-time machines in closed monoidal categories. {I}. J.
  Comput. Syst. Sci.  10(1),  1--43 (1975)

\bibitem{Gramlich03probabilisticand}
Gramlich, G.: Probabilistic and nondeterministic unary automata. In: Proc. of
  Math. Foundations of Computer Science, Springer, LNCS 2747, 2003. pp.
  460--469. Springer (2003)

\bibitem{GratzerGeneralLattice1998}
Gr{\"a}tzer, G.: General Lattice Theory. Birkh{\"a}user Verlag, 2. edn. (1998)

\bibitem{NDLowBoundsHard2006}
Gruber, H., Holzer, M.: Finding lower bounds for nondeterministic state
  complexity is hard. In: Ibarra, O.H., Dang, Z. (eds.) Developments in
  Language Theory: 10th International Conference, DLT 2006, Santa Barbara, CA,
  USA, June 26-29, 2006. Proceedings. pp. 363--374. Springer Berlin Heidelberg,
  Berlin, Heidelberg (2006)

\bibitem{hk71}
Horn, A., Kimura, N.: The category of semilattices. Algebra Univ.  1,  26--38
  (1971)

\bibitem{HROMKOVIC2002202}
Hromkovič, J., Seibert, S., Karhumäki, J., Klauck, H., Schnitger, G.:
  Communication complexity method for measuring nondeterminism in finite
  automata. Information and Computation  172(2),  202--217 (2002),
  \url{http://www.sciencedirect.com/science/article/pii/S089054010193069X}

\bibitem{irs11}
Izhakian, Z., Rhodes, J., Steinberg, B.: Representation theory of finite
  semigroups over semirings. Journal of Algebra  336(1),  139--157 (2011)

\bibitem{MinNfaUnary91}
Jiang, T., McDowell, E., Ravikumar, B.: The structure and complexity of minimal
  nfa’s over a unary alphabet. International Journal of Foundations of
  Computer Science  02(02),  163--182 (1991)

\bibitem{MinNfaHard1993}
Jiang, T., Ravikumar, B.: Minimal {NFA} problems are hard. SIAM Journal on
  Computing  22(6),  1117--1141 (1993)

\bibitem{ContextJipsen2012}
Jipsen, P.: Categories of algebraic contexts equivalent to idempotent semirings
  and domain semirings. In: Kahl, W., Griffin, T.G. (eds.) Relational and
  Algebraic Methods in Computer Science. pp. 195--206. Springer Berlin
  Heidelberg, Berlin, Heidelberg (2012)

\bibitem{StoneSpaces}
Johnstone, P.T.: Stone spaces. Cambridge University Press (1982)

\bibitem{KamedaWeiner1970}
{Kameda}, T., {Weiner}, P.: On the state minimization of nondeterministic
  finite automata. IEEE Transactions on Computers  C-19(7),  617--627 (1970)

\bibitem{TrellisMFCCodes96}
{Kschischang}, F.R.: The trellis structure of maximal fixed-cost codes. IEEE
  Transactions on Information Theory  42(6),  1828--1838 (1996)

\bibitem{MinNfaBiRFSA2009}
Latteux, M., Roos, Y., Terlutte, A.: Minimal {NFA} and {biRFSA} languages.
  RAIRO - Theoretical Informatics and Applications  43(2),  221–237 (2009)

\bibitem{ExtremalLat92}
Markowsky, G.: Primes, irreducibles and extremal lattices. Order  9,  265--290
  (09 1992)

\bibitem{myers2020nondeterministic}
Myers, R.S.R.: Nondeterministic automata and {JSL}-dfas. CoRR  abs/2007.06031
  (2020), \url{https://arxiv.org/abs/2007.06031}

\bibitem{myers2020representing}
Myers, R.S.R.: Representing semilattices as relations. CoRR  abs/2007.10277
  (2020), \url{https://arxiv.org/abs/2007.10277}

\bibitem{mamu15}
Myers, R.S.R., Ad\'amek, J., Milius, S., Urbat, H.: Coalgebraic constructions
  of canonical nondeterministic automata. Theoretical Computer Science  604,
  81--101 (2015)

\bibitem{pin20}
Pin, J.{\'E}.: Mathematical foundations of automata theory (September 2020),
  available at \url{http://www.liafa.jussieu.fr/~jep/PDF/MPRI/MPRI.pdf}

\bibitem{ReversibleAutomata92}
Pin, J.E.: On reversible automata. In: Simon, I. (ed.) LATIN '92. pp. 401--416.
  Springer Berlin Heidelberg, Berlin, Heidelberg (1992)

\bibitem{Polak2001}
Pol{\'a}k, L.: Syntactic semiring of a language. In: Sgall, J., Pultr, A.,
  Kolman, P. (eds.) Mathematical Foundations of Computer Science 2001: 26th
  International Symposium, MFCS 2001 Mari{\'a}nsk{\'e} L{\'a}zne, Czech
  Republic, August 27--31, 2001 Proceedings. pp. 611--620. Springer Berlin
  Heidelberg, Berlin, Heidelberg (2001)

\bibitem{BlockCodesFiniteAut2003}
Shankar, P., Dasgupta, A., Deshmukh, K., Rajan, B.: On viewing block codes as
  finite automata. Theoretical Computer Science  290(3),  1775--1797 (2003)

\bibitem{tamm16}
Tamm, H.: New interpretation and generalization of the {K}ameda-{W}einer
  method. In: Chatzigiannakis, I., Mitzenmacher, M., Rabani, Y., Sangiorgi, D.
  (eds.) ICALP 2016, Rome, Italy. LIPIcs, vol.~55, pp. 116:1--116:12. Schloss
  Dagstuhl - Leibniz-Zentrum f{\"{u}}r Informatik (2016)

\bibitem{TAMM2004135}
Tamm, H., Ukkonen, E.: Bideterministic automata and minimal representations of
  regular languages. Theoretical Computer Science  328(1),  135--149 (2004)

\bibitem{vhmss19}
{van Heerdt}, G., Moerman, J., Sammartino, M., Silva, A.: A (co)algebraic
  theory of succinct automata. Journal of Logical and Algebraic Methods in
  Programming  105,  112--125 (2019)

\end{thebibliography}

\clearpage\appendix

\section{Appendix}
This Appendix provides full proofs and additional details on the examples omitted for space reasons.

\section*{Proof of \autoref{lem:nfarev}}

Let $N=(Q,\delta,I,F)$. We claim that the semilattice isomorphism 
\[ h\colon [\Pow(Q)]^\op\xra{\cong} \Pow(Q), \quad X\mapsto \ol{X}=Q\setminus X,\] 
gives an isomorphism of $\JSL$-dfas from $[\Pow(N)]^\op$ to $\Pow(\rev{N})$.

\medskip \noindent \emph{Preservation of the initial state.} The initial state of $[\Pow(N)^\op]$ is $\ol{F}$, the largest non-final state of $\Pow(N)$. Thus $h$ maps it to $F$, the initial state of $\Pow(\rev{N})$.

\medskip \noindent \emph{Preservation of final states.}
By definition, a state $X$ is final in $[\Pow(N)]^\op$ iff $I\not\seq X$. This is equivalent to $h(X)\cap I \neq \emptyset$, i.e. to $h(X)$ being final in $\Pow(\rev{N})$.

\medskip \noindent \emph{Preservation of transitions.} Let $X,Y\in \Pow(Q)$ and $a\in \Sigma$ such that $X\xra{a} Y$ is a transition in $[\Pow(N)]^\op$. By definition, $Y$ is the set of all $q\in Q$ with $\delta_a[q]\seq X$. Thus, $\ol{Y}$ is the set of all $q\in Q$ such $\delta_a[q]\cap \ol{X} \neq \emptyset$. This means that $\ol{X}\xra{a}\ol{Y}$ is a transition in $\Pow(\rev{N})$.

\section*{Proof of \autoref{lem:astarprops}}

\begin{enumerate}
\item Let $g\colon S\to 2$ be the semilattice morphism corresponding to the prime filter $G=\{x\in S: x\not\leq_S s\}$. Then, for any word $w=a_1\ldots a_n$ in $\Sigma^*$, we have $\delta_{\rev{w}}(s_0)\not\leq_S s$ iff the morphism
\[ 2\xto{i} S\xra{\delta_{a_n}} S \cdots S \xra{\delta_{a_1}} S \xto{g} 2 \]
is equal to $\id\colon 2\to 2$. This is the case iff the dual morphism
\[ 2\xto{g^\ast} S^\op\xra{\delta_{a_1}^\ast} S^\op \cdots S^\op \xra{\delta_{a_n}^\ast} S^\op \xto{i^\ast} 2 \]
is equal to $\id\colon 2\to 2$. Since $g^\ast$ maps $1$ to $s$, this means precisely that the state $s$ of $A^\op$ accepts $w$.
\item follows from part (1) by choosing $s$ to be the initial state of $A^\op$, i.e.~the largest non-final state of $A$.
\item follows via duality: the smallest subautomaton $\reach{A}$ of $A$ dualizes to the smallest quotient automaton $\simple{A^\op}$ of $A^\op$.
\end{enumerate}

\section*{Proof of \autoref{prop:lqdual}}
By \autoref{lem:astarprops}, the dual of a minimal $\JSL$-dfa accepting $\rev{L}$ is a minimal $\JSL$-dfa accepting ${L}$. Thus, by the uniqueness of minimal automata, the unique $\JSL$-automata morphism from  $[\LQ{\rev{L}}]^\op$ to $\LQ{{L}}$, mapping the state $K$ of $[\LQ{\rev{L}}]^\op$ to the language $L([\LQ{\rev{L}}]^\op,K)$ it accepts, is an isomorphism. It only remains to verify that this language is equal to $(\ol{\rev{K}})^{-1}{L}$. To this end, we compute for all $w\in\Sigma^*$:
\begin{align*}
w\in L([\LQ{\rev{L}}]^\op,K) & \iff (\rev{w})^{-1}\rev{L}\not\seq K & \text{by \autoref{lem:astarprops}(1)} \\
& \iff \exists x\in \ol{K}: \rev{w}x\in \rev{L} & \\
& \iff \exists y\in \ol{\rev{K}}: yw\in {L} & \\
& \iff w\in (\ol{\rev{K}})^{-1}{L} & 
\end{align*}

\section*{Proof of \autoref{thm:dependency}}
 
  \begin{enumerate}
    \item We need to show that
    \[
      \alpha\colon \jslLQ{L} \to (\{ \rDR{L}[X] : X \subseteq \LD{L} \}, \cup, \emptyset),
      \qquad
      \alpha(K):= \{ v^{-1} \rev{L} : v \in \rev{K} \,\},
    \]
is an isomorphism. To this end, let $K\in \jslLQ{L}$, say $K=K_1\cup \dots \cup K_n$ for $K_i\in \LW{L}$. We show that  
\[ \alpha(K) = \rDR{L}(\{K_1,\ldots,K_n\}), \]
which immediately implies that $\alpha$ is a well-defined isomorphism of semilattices. To this end, we compute for all $v\in \Sigma^*$:
\begin{align*}
v^{-1}\rev{L}\in \alpha(K) & \iff v\in \rev{K} \\
& \iff \exists i: v\in \rev{K_i} \\
& \iff \exists i: \rev{v} \in K_i \\
& \iff \exists i: v^{-1}\rev{L} \in \rDR{L}[K_i] \\
& \iff v^{-1}\rev{L} \in \rDR{L}(\{K_1,\ldots,K_n\})
\end{align*}
    \item Let us first note that the isomorphism $\dr_L$ from \autoref{prop:lqdual} has the following alternative description:
  \begin{equation}\label{eq:drl}
    \dr_L(U^{-1} {\rev{L}}) = \bigcup \{ K \in \LW{{L}} : K \cap \rev{U} = \emptyset \}
    \qquad
    \text{for every $U \subseteq \Sigma^*$}.
  \end{equation}
In fact, for every $w\in \Sigma^*$ we compute:
\begin{align*}
w\in \dr_L(U^{-1}\rev{L}) & \iff w\in \ol{\rev{(U^{-1}\rev{L})}}^{-1}{L} & \text{def. $\dr_L$} \\
& \iff \exists v\in \ol{U^{-1}\rev{L}}: \rev{v}w\in {L} & \\
& \iff \exists v\in \Sigma^*: [\rev{v}w\in {L} \wedge \forall u\in U: uv\not\in \rev{L}] \\
& \iff \exists y\in \Sigma^*: [yw\in {L} \wedge \forall u\in U: y\rev{u}\not\in {L}] \\
& \iff \exists y\in \Sigma^*: [w\in y^{-1}{L} \wedge y^{-1}{L}\cap \rev{U} = \emptyset] \\
& \iff w\in \bigcup \{ K \in \LW{{L}} : K \cap \rev{U} = \emptyset \}.
\end{align*}
It thus follows for all $u,v\in \Sigma^*$:
\begin{align*}
 u^{-1}L\not\seq \dr_L(v^{-1}\rev{L}) & \iff u^{-1}L \not \in \{ K \in \LW{{L}} : \rev{v}\not \in K \} \\
& \iff \rev{v}\in u^{-1}L \\
& \iff \rDR{L}(u^{-1}L,v^{-1}\rev{L}).
\end{align*}
\item follows immediately from (2), restricted to $\rRDR{L}$.
  \end{enumerate}

\section*{Proof of \autoref{prop:blqdual}}

\begin{enumerate}
\item Let $\At{x}$ denote the unique atom of $\jslBLQ{L}$ containing the word $x\in \Sigma^*$.
For any $v,w\in \Sigma^*$ we have $v^{-1}\rev{L} = {w}^{-1}\rev{L}$ iff $\At{\rev{v}}=\At{\rev{w}}$. In fact,
\begin{align*}
& v^{-1}\rev{L} = {w}^{-1}\rev{L} \\
\text{iff}\quad & \forall x\in \Sigma^*: vx\in \rev{L} \iff wx\in \rev{L} \\
\text{iff}\quad & \forall y\in \Sigma^*: \rev{v}\in y^{-1}L \iff \rev{w}\in y^{-1}L \\
\text{iff}\quad & \At{\rev{v}} = \At{\rev{w}}.
\end{align*}
In the final step, we use that the boolean algebra $\jslBLQ{L}$ is generated by the left derivatives of $L$, so two words belong to the same atom iff they belong to the same left derivatives.
\item It follows that the map $h\colon \Pow(\minDfa{\rev{L}})\to [\jslBLQ{L}]^\op$ defined by
\[ \{\,w_1^{-1}\rev{L},\ldots, w_n^{-1}\rev{L} \,\} \quad\mapsto \quad \bigcap_{i=1}^n \overline{\At{\rev{w_i}}} \]
gives a well-defined isomorphism of semilattices. It remains to prove that it is an automata morphism.

\medskip \noindent \emph{Preservation of the initial state.} The initial state $\{\rev{L}\}$ of $\Pow(\minDfa{\rev{L}})$ is mapped to $\ol{\At{\epsilon}}$.  This is the largest non-final state of $\jslBLQ{L}$, i.e.~the initial state of $[\jslBLQ{L}]^\op$.

\medskip\noindent\emph{Preservation of final states.} Recall that the final states of $[\jslBLQ{L}]^\op$ are those languages in $\jslBLQ{L}$ not containing $L$. Thus,
\begin{align*}
& \{\,w_1^{-1}\rev{L},\ldots, w_n^{-1}\rev{L} \,\} \text{ final in $\jslBLQ{L}$} \\
\text{iff}\quad & \text{$w_i\in \rev{L}$ for some $i$} \\
\text{iff}\quad & \text{$\rev{w_i}\in L$ for some $i$} \\
\text{iff}\quad & \text{$\At{\rev{w_i}}\seq L$ for some $i$} \\
\text{iff}\quad & \text{$L \not\seq \overline{\At{\rev{w_i}}}$ for some $i$} \\
\text{iff}\quad & L \not\seq \bigcap_{i=1}^n\overline{\At{\rev{w_i}}} \\
\text{iff}\quad & \bigcap_{i=1}^n\overline{\At{\rev{w_i}}} \text{ final in $[\jslBLQ{L}]^\op$} 
\end{align*}
\end{enumerate} 

\medskip\noindent\emph{Preservation of transitions.} Since the semilattice $\Pow(\minDfa{\rev{L}})$ is generated by the left derivatives of $\rev{L}$, it suffices to prove that for each $w\in \Sigma^*$ and $a\in \Sigma$ we have the transition 
\[h( \{  w^{-1}\rev{L}\} )~ \xra{a} ~h( \{ a^{-1}w^{-1}\rev{L} \}), \]
i.e.~
\[ \ol{\At{\rev{w}}}~ \xra{a}~ \ol{\At{a\rev{w}}} \]
in $[\jslBLQ{L}]^\op$. But this is immediate because $a^{-1}\At{a\rev{w}} \supseteq \At{\rev{w}}$.

\section*{Proof of \autoref{prop:blrqdual}}
The proof is much analogous to the one of \autoref{prop:blqdual}.
\begin{enumerate}
\item Let $\At{x}$ denote the atom of $\jslBLRQ{L}$ containing the word $x\in \Sigma^*$.
For any two words $v,w\in \Sigma^*$ we have $v\equiv_{\rev{L}} w$ iff $\At{\rev{v}}=\At{\rev{w}}$. In fact,
\begin{align*}
& v\equiv_{\rev{L}} w \\
\text{iff}\quad & \forall x,y\in \Sigma^*: v\in x^{-1}\rev{L}y^{-1} \iff w\in x^{-1}\rev{L}y^{-1} \\
\text{iff}\quad & \forall s,t\in \Sigma^*: \rev{v}\in s^{-1}Lt^{-1} \iff \rev{w}\in s^{-1}Lt^{-1} \\
\text{iff}\quad & \At{\rev{v}} = \At{\rev{w}}.
\end{align*}
In the final step, we use that the boolean algebra $\jslBLRQ{L}$ is generated by the two-sided derivatives of $L$, so two words belong to the same atom iff they belong to the same two-sides derivatives.
\item It follows that the map $h\colon \Pow(\Syn{\rev{L}})\to [\jslBLRQ{L}]^\op$ defined by
\[ \{\,[w_1]_{\rev{L}},\ldots, [w_n]_{\rev{L}} \,\} \quad\mapsto \quad \bigcap_{i=1}^n \overline{\At{\rev{w_i}}} \]
gives a well-defined isomorphism of semilattices. It remains to prove that it is an automata morphism.

\medskip \noindent \emph{Preservation of the initial state.} The initial state $\{ [\epsilon]_{\rev{L}} \}$ of $\Pow(\Syn{\rev{L}})$ is mapped to $\ol{\At{\epsilon}}$.  This is the largest non-final state of $\jslBLRQ{L}$, i.e.~the initial state of $[\jslBLRQ{L}]^\op$.

\medskip\noindent\emph{Preservation of final states.} The final states of $[\jslBLRQ{L}]^\op$ are those languages in $\jslBLRQ{L}$ not containing $L$. Thus,
\begin{align*}
& \{\, [w_1]_{\rev{L}},\ldots, [w_n]_{\rev{L}} \,\} \text{ final in $\jslBLRQ{L}$} \\
\text{iff}\quad & \text{$w_i\in \rev{L}$ for some $i$} \\
\text{iff}\quad & \text{$\rev{w_i}\in L$ for some $i$} \\
\text{iff}\quad & \text{$\At{\rev{w_i}}\seq L$ for some $i$} \\
\text{iff}\quad & \text{$L \not\seq \overline{\At{\rev{w_i}}}$ for some $i$} \\
\text{iff}\quad & L \not\seq \bigcap_{i=1}^n\overline{\At{\rev{w_i}}} \\
\text{iff}\quad & \bigcap_{i=1}^n\overline{\At{\rev{w_i}}} \text{ final in $[\jslBLRQ{L}]^\op$} 
\end{align*}
\end{enumerate} 

\medskip\noindent\emph{Preservation of transitions.} Since the semilattice $\Pow(\Syn{\rev{L}})$ is generated by the elements of $\Syn{L}$, it suffices to prove that for each $w\in \Sigma^*$ and $a\in \Sigma$ we have the transition 
\[h( \{  [w]_{\rev{L}}  \} )~ \xra{a} ~h( \{ [wa]_{\rev{L}} \}), \]
i.e.~
\[ \ol{\At{\rev{w}}}~ \xra{a}~ \ol{\At{a\rev{w}}} \]
in $[\jslBLRQ{L}]^\op$. But this is immediate because $a^{-1}\At{a\rev{w}} \supseteq \At{\rev{w}}$.

\section*{Proof of \autoref{prop:tsdual}}
Let $A=(S,\delta,s_0,F)$. For any $K\seq \Sigma^*$ we put $\delta_K := \bigvee_{w\in K} \delta_{w}$.
\begin{enumerate}
\item We first show that 
\begin{equation}\label{eq:langtsdual}
L([\ts{A}]^\op,\delta_K) \;=\; \bigcup_{v\in\Sigma^*} L(A^\op,\delta_{vK}(s_0)) (\rev{v})^{-1}\qquad\text{for each $K\seq \Sigma^*$}.
\end{equation} 
To see this, we compute for all $u\in \Sigma^*$:
\begin{align*}
& u\in L([\ts{A}]^\op,\delta_K) &  \\
\text{iff}\quad & \delta_{\rev{u}}\not\leq \delta_K & \text{by \autoref{lem:astarprops}(1)} \\
\text{iff}\quad & \exists v\in \Sigma^*:\, \delta_{\rev{u}}(\delta_v(s_0)) \not\leq_S \delta_K(\delta_v(s_0))  & \text{since $A$ is reachable} \\
\text{iff}\quad & \exists v\in \Sigma^*:\, \delta_{v\rev{u}}(s_0)\not \leq_S \delta_{vK}(s_0) & \\
\text{iff}\quad & \exists v\in \Sigma^*:\, u\rev{v}\in L(A^\op,\delta_{vK}(s_0)) & \text{by \autoref{lem:astarprops}(1)} \\
\text{iff}\quad & \exists v\in \Sigma^*:\, u\in L(A^\op,\delta_{vK}(s_0))(\rev{v})^{-1}. &
\end{align*}\clearpage
\item For any $w\in \Sigma^*$, consider the two semilattice morphisms
\begin{align*} \gamma_w\colon \ts{A}\to \ts{A},&\quad f\mapsto \delta_w\o f, \\
\varphi_w\colon \ts{A}\to \ts{A},&\quad f\mapsto f\o \delta_w. 
\end{align*}
along with their dual morphisms $\gamma_w^\ast, \varphi_w^\ast\colon [\ts{A}]^\op\to [\ts{A}]^\op$. We claim that
\begin{equation}\label{eq:rqc}
L([\ts{A}]^\op, \delta_K)(\rev{w})^{-1} = L([\ts{A}]^\op, \varphi_w^*(\delta_K))  \qquad\text{for each $K\seq \Sigma^*$}.
\end{equation}
To see this, we compute as follows for all $u\in \Sigma^*$, where $\leq$ is the order of the semilattice $\JSL(S,S)$:
\begin{align*}
& u\in L([\ts{A}]^\op, \varphi_w^*(\delta_K)) & \\
\text{iff}\quad & \id_S \not\leq (\gamma_{\rev{u}})^*(\varphi_w^*(\delta_K)) & \text{def. $L(\dash,\dash)$}  \\
\text{iff}\quad & \id_S\not\leq (\varphi_w\o \gamma_{\rev{u}})^\ast(\delta_K) & \\
\text{iff}\quad & \phi_w\o \gamma_{\rev{u}}(\id_S) \not\leq \delta_K & \text{by adjointness} \\
\text{iff}\quad & \delta_{w\rev{u}} \not\leq \delta_K &  \\
\text{iff}\quad & \gamma_{w\rev{u}}(\id_S) \not\leq \delta_K & \\
\text{iff}\quad & \id_S \not\leq (\gamma_{w\rev{u}})^\ast(\delta_K) & \text{by adjointness} \\
\text{iff}\quad & u\rev{w}\in L([\ts{A}]^\op, \delta_K) & \text{def. $L(\dash,\dash)$}\\
\text{iff}\quad & u\in L([\ts{A}]^\op, \delta_K)(\rev{w})^{-1}
\end{align*}
\item We are ready to prove the proposition. Since both $[\ts{A}]^\op$ and $\rqc{A^\op}$ are simple $\JSL$-dfas, and thus can be viewed as subautomata of $\Fin{L}$, it suffices to show that they contain the same languages. The inclusion $[\ts{A}]^\op\seq \rqc{A^\op}$ follows from \eqref{eq:langtsdual}. For the reverse inclusion, since $[\ts{A}]^\op$ is closed under right derivatives by \eqref{eq:rqc}, we only need to prove that $A^\op \seq [\ts{A}]^\op$. To this end, we show that, for any $s\in S$,
\[ L(A^\op, s) = L([\ts{A}]^\op, \delta_K),\quad\text{where}\quad K = \{\,w\in \Sigma^*\;:\; \delta_w(s_0)\leq_S s\,\}. \]
For the proof, we first note that for all $u\in \Sigma^*$,
\begin{equation}
\label{eq:deltavK} \delta_u(s_0)\leq_S s \quad\iff\quad \forall v\in \Sigma^*: \delta_{vu}(s_0)\leq_S \delta_{vK}(s_0). 
\end{equation}
In fact, ``$\Leftarrow$'' follows by taking $v=\epsilon$; we have $s=\delta_K(s_0)$ because $A$ is reachable. For ``$\To$'', suppose that $\delta_u(s_0)\leq_S s$. Then $u\in K$ and therefore 
\[ \delta_{vu}(s_0) \leq_S \bigvee_{w\in K} \delta_{vw}(s_0) = \delta_{vK}(s_0) \]
 We now compute
\begin{align*}
& u\in L(A^\op,s)  & \\
\text{iff}\quad &  \delta_{\rev{u}}(s_0)\not\leq_S s & \text{by \autoref{lem:astarprops}(1)} \\
 \text{iff}\quad & \exists v\in \Sigma^*:\,\delta_{v\rev{u}}(s_0)\not\leq_S \delta_{vK}(s_0) & \text{by \eqref{eq:deltavK}} \\
 \text{iff}\quad & \exists v\in \Sigma^*:\, u\rev{v} \in L(A^\op,\delta_{vK}(s_0)) & \text{by \autoref{lem:astarprops}(1)}  \\
 \text{iff}\quad & \exists v\in \Sigma^*:\, u\in L(A^\op,\delta_{vK}(s_0))(\rev{v})^{-1} & \\
 \text{iff}\quad & u\in L([\ts{A}]^\op, \delta_K) & \text{by \eqref{eq:langtsdual}}
\end{align*}
This concludes the proof.
\end{enumerate}

\section*{Proof of \autoref{thm:ns_char}}

Let $d(L)$ denote the least degree of any boolean representation extending the canonical representation $\kappa_L\o \mu_L$.
\begin{enumerate}
\item A boolean presentation of $\Sigma^*$ is given by a finite semilattice lattice $S$ together with a family of semilattice morphisms $\delta=(\delta_a\colon S\to S)_{a\in \Sigma}$. An equivariant map between boolean presentations $(S,\delta)$ and $(S',\delta')$ is a semilattice morphism $h\colon S\to S'$ with $\delta_a'\o h = h\o\delta_a$ for all $a\in \Sigma$. If $S$ carries a $\JSL$-automata structure $(S,\delta,i,f)$ and $h$ is a monic, there exists an automata structure on $S'$ making $h$ an automata morphism: put $i':=h\o i$, and choose $f'\colon S'\to 2$ to be any semilattice morphism with $f'=h\o f$. Such an $f'$ exists because the semilattice $2$ is an injective object of $\JSL$.
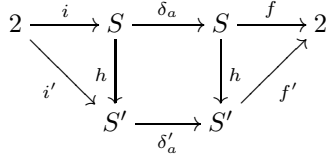
\[
\xymatrix{
2 \ar[r]^{i} \ar[dr]_{i'} & S \ar[r]^{\delta_a} \ar[d]_h & S \ar[d]^h \ar[r]^f & 2 \\
 & S' \ar[r]_{\delta_a'} & S' \ar[ur]_{f'} & 
}
\]
\item To prove $d(L)\leq \ns{L}$, suppose that $N$ is an nfa accepting the language $L$. Consider the $\JSL$-subautomaton $\jslLangs{N}=\simple{\Pow(N)}$ of $\Fin{L}$ carried by the semilattice of all languages accepted by subsets of $N$. Note that $\LQ{L}$ is a subautomaton of  $\jslLangs{N}$: every finite union $\bigcup_{i}w_i^{-1}L$ of left derivatives of $L$ is accepted by the set of all states of $N$ reachable on input $w_i$ for some $i$. Thus, the inclusion map $\LQ{L}\monoto \jslLangs{N}$ defines an extension of the canonical representation $\kappa_L\o \mu_L$. Since the semilattice $\jslLangs{N}$ is generated by the set of languages accepted by single states of $N$, it follows that the degree of this representation is at most the number of states of $N$.
\item To prove $\ns{L}\leq d(L)$, suppose that $(S,\delta)$ is a boolean representation of $\Sigma^*$ of degree $k$ extending $\kappa_L\o \mu_L$, witnessed by an injective equivariant map $h\colon \LQ{L}\monoto S$. By part (1), we can equip $S$ with a $\JSL$-dfa structure making $h$ an automata morphism. Since morphisms preserve accepted languages, it follows that $S$ accepts $L$. The automaton $S$ has $k$ join-irreducibles, so \autoref{rem:jsldfa_vs_nfa} shows that there exists an nfa on $k$ states accepting $L$.
\end{enumerate}

\section*{Proof of \autoref{thm:atomic_nfa_char}}

\begin{rem}\label{rem:freejsldfa}
The subset construction, restricted to dfas, gives rise to a left adjoint $\Pow\colon \Aut{\Set_\f} \to \Aut{\JSL_\f}$ between the categories of dfas and $\JSL$-dfas. Thus, for any dfa $D$ and any $\JSL$-dfa $A$, there is a bijective correspondence between dfa morphisms from $D$ to $A$ and $\JSL$-dfa morphisms from $\Pow(D)$ to $A$.
\end{rem}
Our proof of \autoref{thm:atomic_nfa_char} is essentially an instance of the self-duality of $\JSL$-dfas. Let $L$ be the language accepted by $N$. We establish the theorem by showing that each of the following statements is equivalent to the next one:
\begin{enumerate}
\item $N$ is atomic.
\item There exists a $\JSL$-automata morphism from $\Pow(N)$ to $\jslBLQ{L}$. 
\item There exists a $\JSL$-automata morphism from $\Pow(\minDfa{\rev{L}})$ to ${\Pow(\rev{N})}$.
\item There exists a dfa morphism from $\minDfa{\rev{L}}$ to ${\Pow(\rev{N})}$.
\item There exists a dfa morphism from $\minDfa{\rev{L}}$ to $\rsc{\rev{N}}$.
\item $\rsc{\rev{N}}$ is a minimal dfa.
\end{enumerate}

\medskip\noindent \emph{Ad (1)$\Lra$(2).} The unique automata morphism $m_{\Pow(N)}\colon \Pow(N)\to \Fin{L}$ maps every state of $\Pow(N)$ to the language it accepts. Thus, $N$ is atomic iff $m_{\Pow(N)}$ factorizes through the subautomaton $\jslBLQ{L}$ of $\Fin{L}$.

\medskip\noindent \emph{Ad (2)$\Lra$(3).} This follows via duality from \autoref{lem:nfarev}, \autoref{lem:astarprops} and \autoref{prop:blqdual}. 

\medskip\noindent \emph{Ad (3)$\Lra$(4).} This follows from \autoref{rem:freejsldfa}.

\medskip\noindent \emph{Ad (4)$\Lra$(5).} Since $\minDfa{\rev{L}}$ is a reachable dfa, every dfa morphism from $\minDfa{\rev{L}}$ to ${\Pow(\rev{N})}$ factorizes through the dfa-reachable part $\rsc{\rev{N}}$ of $\Pow(\rev{N})$.

\medskip\noindent \emph{Ad (5)$\Lra$(6).} Every dfa morphism from $\minDfa{\rev{L}}$ to $\rsc{\rev{N}}$ is an isomorphism: it is injective because $\minDfa{\rev{L}}$ is a simple dfa and surjective because $\rsc{\rev{N}}$ is a reachable dfa. Conversely, if $\rsc{\rev{N}}$ is a minimal dfa, then it is isomorphic to $\minDfa{\rev{L}}$ by the uniqueness of minimal dfas.

\section*{Proof of \autoref{thm:subatomic_nfa_char}}
Let us first recall the concept of algebraic language recognition~\cite{pin20}.
\begin{rem}\label{rem:algrec}
 A finite monoid $M$ is said to \emph{recognize} the language $L\seq\Sigma^*$ if there exists a monoid morphism $h\colon \Sigma^*\to M$ and a subset $P\seq M$ with $L=h^{-1}[P]$. Regular languages are exactly the languages recognizable by finite monoids. In fact, we have the following connections between monoids and dfas:
\begin{enumerate}
\item If $L$ is recognized by a finite monoid $M$ via $h\colon \Sigma^*\to M$ and $P\seq M$, then $M$ can be viewed as dfa accepting $L$, with transitions $m\xra{a}m\bullet h(a)$ for $m\in M$ and $a\in \Sigma$, initial state $1_M$, and final states $P$.
\item Conversely, if $L$ is accepted by a dfa $D=(S,\delta,s_0,F)$, then the transition monoid $\tm{D}$ recognizes $L$ via the morphism $h\colon \Sigma^*\epito \tm{D}$, $w\mapsto \delta_w$, and $P=\{\delta_w\colon w\in L\}$. In particular, the syntactic monoid recognizes $L$ via the syntactic morphism $\mu_L\colon \Sigma^*\epito \Syn{L}$. It can be characterized as the least quotient monoid of $\Sigma^*$ recognizing $L$: for any surjective monoid morphism $h\colon \Sigma^*\epito M$ recognizing $L$, there is a unique morphism $g\colon M\epito \Syn{L}$ with $\mu_L = g\o h$:
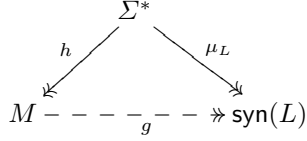
\[
\xymatrix{
& \Sigma^* \ar@{->>}[dl]_h \ar@{->>}[dr]^{\mu_L} & \\
M \ar@{-->>}[rr]_g & &  \Syn{L}
}
\] 
\item Finally, there is a tight connection between morphisms of monoids and dfas. Suppose that two surjective monoid morphisms $h_i\colon \Sigma^*\epito M_i$ and subsets $P_i\seq M_i$ for $i=1,2$ are given. As in part (1), we view $M_1$ and $M_2$ as dfas. Then every dfa morphism $g\colon M_1\to M_2$ makes the triangle below commute:
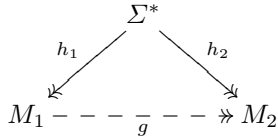
\[
\xymatrix{
& \Sigma^* \ar@{->>}[dl]_{h_1} \ar@{->>}[dr]^{h_2} & \\
M_1 \ar@{-->>}[rr]_g & &  M_2
}
\] 
In fact, $M_1$ and $M_2$ accept the same language $L$ and $\Sigma^*$ can be seen as the initial dfa accepting $L$ when equipped with $L\seq \Sigma^*$ as the set of final states. From the surjectivity of $h_1$ it easily follows that $g$ is a monoid morphism. Conversely, every monoid morphism $g$ making the above triangle commute and satisfying $g[P_1]=P_2$ is a dfa morphism.
\end{enumerate}
\end{rem}

\begin{rem}\label{rem:tsreachable}
For any $\JSL$-dfa $A$, the dfa-reachable part of $\ts{\reach{A}}$ is $\tm{A_r}$, where $A_r$ denotes the dfa-reachable part of $A$. In fact, letting $\reach{A}=(S,\delta,s_0,F)$ and $A_r=(S_r,\delta_r, s_{0,r}, F_r)$, we have that $A_r$ is a sub-dfa of $\reach{A}$. Then the map $(\delta_r)_w\mapsto \delta_{w}$ gives a well-defined injective dfa morphism from $\tm{A_r}$ to $\ts{\reach{A}}$, using that the semilattice $S$ is generated by the subset $S_r\seq S$. Thus, $\tm{A_r}$ is a sub-dfa of $\ts{\reach{A}}$. Since it is reachable, it it isomorphic to the dfa-reachable part of $\ts{\reach{A}}$.
\end{rem}
With these preparations, we are ready to prove \autoref{thm:subatomic_nfa_char}. Again, the argument crucially rests on the self-duality of $\JSL$-dfas. We show that each of the following statements is equivalent to the next one:
\begin{enumerate}
\item $N$ is subatomic.
\item There exists a $\JSL$-dfa morphism from $\Pow(N)$ to $\jslBLRQ{L}$.
\item There exists a $\JSL$-dfa morphism from $\rqc{\simple{P(N)}}$ to $\jslBLRQ{L}$.
\item There exists a $\JSL$-dfa morphism from $\Pow(\Syn{\rev{L}})$ to $\ts{\reach{\Pow(\rev{N})}}$.
\item There exists a dfa morphism from $\Syn{\rev{L}}$ to $\ts{\reach{\Pow(\rev{N})}}$.
\item There exists a dfa morphism from $\Syn{\rev{L}}$ to $\tm{\rsc{\rev{N}}}$.
\item The monoids $\Syn{\rev{L}}$ and $\tm{\rsc{\rev{N}}}$ are isomorphic.
\end{enumerate}
\medskip\noindent \emph{Ad (1)$\Lra$(2).} The unique automata morphism $m_{\Pow(N)}\colon \Pow(N)\to \Fin{L}$ maps every state of $\Pow(N)$ to the language it accepts. Thus, $N$ is subatomic iff $m_{\Pow(N)}$ factorizes through the subautomaton $\jslBLRQ{L}$ of $\Fin{L}$.

\medskip\noindent \emph{Ad (2)$\Lra$(3).} This is clear since $\jslBLRQ{L}$ is closed under right derivatives.

\medskip\noindent \emph{Ad (3)$\Lra$(4).} This follows via duality from \autoref{lem:nfarev}, \autoref{prop:blrqdual} and \autoref{prop:tsdual}.

\medskip\noindent \emph{Ad (4)$\Lra$(5).} This follows from \autoref{rem:freejsldfa}.

\medskip\noindent \emph{Ad (5)$\Lra$(6).} 
Putting $A=\Pow(\rev{N})$ in \autoref{rem:tsreachable}, we see that $\tm{\rsc{\rev{N}}}$ is the dfa-reachable part of $\ts{\reach{\Pow(\rev{N})}}$. Since $\Syn{\rev{L}}$ is reachable as a dfa, it follows that every dfa morphism into $\ts{\reach{\Pow(\rev{N})}}$ factorizes through $\tm{\rsc{\rev{N}}}$.

\medskip\noindent \emph{Ad (6)$\To$(7).} Let $q_{\rev{N}}\colon \Sigma^*\epito \tm{\rsc{\rev{N}}}$ denote the canonical monoid morphism mapping $w\in \Sigma^*$ to the transition morphism $\delta_w$ of the dfa $\rsc{\rev{N}}$. Note that the dfa structure of $\tm{\rsc{\rev{N}}}$ is precisely the one induced by $q_{\rev{N}}$. Thus, given a dfa morphism $h\colon \Syn{\rev{L}} \to \tm{\rsc{\rev{N}}}$ we know that the following diagram commutes by initiality, see \autoref{rem:algrec}(3):
\begin{equation}\label{eq:triangle}
\xymatrix@C-1em{
& \Sigma^* \ar@{->>}[dl]_{\mu_{\rev{L}}} \ar@{->>}[dr]^{q_{\rev{N}}} & \\
\Syn{\rev{L}} \ar@{-->>}[rr]_h & & \tm{\rsc{\rev{N}}} 
}
\end{equation}
Then $h$ is necessarily a monoid morphism because $\mu_{\rev{L}}$ is surjective. 
Since $q_{\rev{N}}$ recognizes the language $\rev{L}$, we get a unique monoid morphism $g\colon \tm{\rsc{\rev{N}}}  \to \Syn{\rev{L}}$ with $g\o q_{\rev{N}} = \mu_{\rev{L}}$. It follows that $h$ is an isomorphism with $h^{-1}=g$.

\medskip\noindent \emph{Ad (7)$\To$(6).} Suppose that the monoids $\Syn{\rev{L}}$ and $\tm{\rsc{\rev{N}}}$ are isomorphic. Let again $g\colon \tm{\rsc{\rev{N}}}  \to \Syn{\rev{L}}$ be the unique monoid morphism with $g\o q_{\rev{N}} = \mu_{\rev{L}}$. Then $g$ is surjective because $\mu_{\rev{L}}$ is. Since $\Syn{\rev{L}}$ and $\tm{\rsc{\rev{L}}}$ have the same number of elements, it follows that $g$ is also injective, i.e.~an isomorphism of monoids. Then \autoref{rem:algrec}(3) shows that its inverse $g^{-1}\colon \Syn{\rev{L}} \to \tm{\rsc{\rev{N}}}$ is a dfa morphism.

\section*{Proof of \autoref{thm:nmu_eq_subatomic}}
Let $a(L)$ denote the least number of states of any subatomic nfa accepting $L$. We are to prove $a(L)=\nsyn{L}$.
\begin{enumerate}
\item To prove $\nsyn{L}\leq a(L)$, suppose that $N$ is a subatomic nfa accepting the language $L$. Consider the subsemilattice $\jslLangs{N}=\simple{\Pow(N)}$ of $\Fin{L}$ of all languages accepted by subsets of $N$. We claim that
\[ \rho\colon \Syn{L}\to \JSL(\jslLangs{N},\jslLangs{N}),\quad [w]_L \mapsto \lambda K.w^{-1}K\]
is a boolean representation of $\Syn{L}$ extending the canonical one. This is obvious once we prove $\rho$ to be a well-defined map, i.e.
\[ v\equiv_L w \quad\text{implies}\quad v^{-1}K=w^{-1}K \]
for $v,w\in\Sigma^*$ and $K\in \jslLangs{N}$. Since $K\in \jslBLRQ{L}$, the boolean algebra generated by all two-sided derivatives of $L$, and derivatives commute with all set-theoretic boolean operations, we can assume w.l.o.g.\ that $K=s^{-1}Lt^{-1}$ for some $s,t\in \Sigma^*$. Then, for all $x\in \Sigma^*$,
\begin{align*}
x\in v^{-1}K &\iff~ x\in v^{-1}s^{-1}Lt^{-1} & \\
& \iff~ svxt\in L &\\
& \iff~ swxt \in L & \text{since $v\equiv_L w$} \\
& \iff~  x\in w^{-1}s^{-1}Lt^{-1} & \\
& \iff~ x\in w^{-1}K &
\end{align*}
proving that $v^{-1}L=w^{-1}L$, as required. Since the semilattice $\jslLangs{N}$ is generated by the set of languages accepted by single states of $N$, it follows that $\deg(\rho)$ is at most the number of states of $N$.
\item To prove $a(L)\leq \nsyn{L}$, let $\rho\colon \Syn{L}\to \JSL(S,S)$ be a boolean representation of $\Syn{L}$ extending the canonical one. Then $\rho\o \mu_L\colon \Sigma^*\to \JSL(S,S)$ extends the canonical presentation $\kappa_L\o \mu_L$ of $\Sigma^*$, and so like in proof of \autoref{thm:ns_char} we can equip $S$ with the structure of a $\JSL$-dfa $A=(S,\delta,i,f)$ accepting $L$. Its extended transition morphism for $w\in \Sigma^*$ is given by 
\[\delta_w\colon S\to S,\quad s\mapsto \rho([w]_L)(s). \]
In particular, $v\equiv_L w$ implies $\delta_v=\delta_w$, which shows that every state of $A$ accepts a union of syntactic congruence classes of $L$. Since
\[ [w]_L = \bigcap_{xwy\in L} x^{-1}Ly^{-1} \cap \bigcap_{xwy\not\in L} \ol{x^{-1}Ly^{-1}},  \]
it follows that all languages accepted by states of $A$ lie in $\jslBLRQ{L}$. Therefore, the nfa $N$ of join-irreducibles of $A$ (see \autoref{rem:jsldfa_vs_nfa}) is a subatomic nfa with $\deg(\rho)$ states accepting $L$.
\end{enumerate}

\section*{Proof of Theorem 5.2}
\begin{enumerate}
\item Suppose that $\Syn{L}=\tm{\minDfa{L}}$ is cyclic. Then there exists $w_0\in\Sigma^*$ such that the map $\lambda X. w_0^{-1} X\colon \LW{L} \to \LW{L}$ generates $\tm{\minDfa{L}}$. We claim that, for all $K,M\seq\Sigma^*$
\begin{equation}\label{eq:derequal} 
K^{-1}L = M^{-1}L \quad\text{iff}\quad [\forall n\in\Nat: w_0^n\in K^{-1}L \iff w_0^n\in M^{-1}L]. \end{equation}
The ``only if'' direction is trivial. For the converse, suppose that $K^{-1}L\neq M^{-1}L$. W.l.o.g. we may assume that there exists $w\in K^{-1}L\setminus M^{-1}L$. Choose $i_1,\ldots i_k$ and $j_1,\ldots, j_m$ such that $K^{-1}L = \bigcup_{p=1}^k (w_0^{i_p})^{-1}L$ and $M^{-1}L= \bigcup_{r=1}^m (w_0^{j_r})^{-1}L$. Moreover, choose $n\in \Nat$ such that $w^{-1}L=(w_0^n)^{-1}L$. Then we have $w\in (w_0^{i_p})^{-1}L$ for some $p$ and thus $w_0^{i_p}\in w^{-1}L=(w_0^n)^{-1}L$, using that $\tm{\minDfa{L}}$ is a commutative monoid. Thus, $w_0^n\in (w_0^{i_p})^{-1}L\seq K^{-1}L$. On the other hand, we have $w\not\in (w_0^{j_r})^{-1}L$ for all $r$, so the same argument shows that $(w_0)^n\not\in M^{-1}L$.

\item Fix an alphabet $\Sigma_0 = \{ a_0 \}$ disjoint from $\Sigma$ and consider the unary language
\[
  L_0 := \{\, a_0^n \;:\; n \in \Nat,\, w_0^n \in L \,\} \subseteq \Sigma_0^*.
\]
 Let $g: \Sigma_0^* \to \Sigma^*$ be the monoid morphism where $g(a_0) := w_0$. We claim that the following map is a $\JSL$-isomorphism:
\[f \colon\jslLQ{L_0} \xra{\cong} \jslLQ{L},\quad f(X^{-1} L_0) := g[X]^{-1} L.\]
To see that $f$ is well-defined and injective, we prove for all $X,Y\seq\Sigma_0^*$:
\[ X^{-1}L_0 = Y^{-1}L_0 \quad \text{iff}\quad g[X]^{-1}L = g[Y]^{-1}L. \]
In fact, we have
\begin{align*}
&  X^{-1}L_0 = Y^{-1}L_0  \\
\text{iff}\quad & \forall n\in\Nat: a_0^n\in X^{-1}L_0 \iff a_0^n\in Y^{-1}L_0  \\
\text{iff}\quad & \forall n\in \Nat: [\exists a_0^k\in X : a_0^{n+k}\in L_0] \iff [\exists a_0^m\in Y : a_0^{n+m}\in L_0] \\
\text{iff}\quad & \forall n\in \Nat: [\exists a_0^k\in X: w_0^{n+k}\in L] \iff [\exists a_0^m\in Y : w_0^{n+m}\in L] \\
\text{iff}\quad & \forall n\in \Nat: [\exists a_0^k\in X: w_0^{n}\in (g(a_0)^k)^{-1}L] \Lra [\exists a_0^m\in Y : w_0^{n}\in (g(a_0)^m)^{-1}L] \\
\text{iff}\quad & \forall n\in \Nat: w_0^{n}\in g[X]^{-1}L \iff w_0^{n}\in g[Y]^{-1}L \\
\text{iff}\quad & g[X]^{-1}L = g[Y]^{-1}L
\end{align*}
where the final step uses \eqref{eq:derequal}. This proves $f$ to be  well-defined and injective. Moreover, it immediately follows from the definition that $f$ is surjective and preserves finite unions.

\item  For each $a \in \Sigma$ choose $n_a\in \Nat$ such that $a^{-1} K = (w_0^{n_a})^{-1} K$ for all $K\in \LW{L}$. The respective transition endomorphisms of the $\JSL$-automata $\jslLQ{L_0}$ and $\jslLQ{L}$ determine each other in the sense that the following diagrams commute:
\[
\xymatrix@R-1em{
\jslLQ{L_0} \ar[r]^f_\cong \ar[d]_{a_0^{-1}(\dash)} & \jslLQ{L} \ar[d]^{ w_0^{-1}(\dash)  }  \\
\jslLQ{L_0} \ar[r]_f^\cong & \jslLQ{L}
}
\qquad
\xymatrix@R-1em{
\jslLQ{L_0} \ar[r]^f_\cong \ar[d]_{(a_0^{n_a})^{-1}(\dash)} & \jslLQ{L} \ar[d]^{ a^{-1}(\dash)  }  \\
\jslLQ{L_0} \ar[r]_f^\cong & \jslLQ{L}
}
\]
It follows that extensions of the canonical representations $\kappa_L$ and $\kappa_L\o \mu_L$ correspond uniquely to extensions of the canonical representations $\kappa_{L_0}$ and $\kappa_{L_0}\o \mu_{L_0}$, respectively. Therefore, $\ns{L} = \ns{L_0}$ by \autoref{thm:ns_char} and $\nsyn{L} = \nsyn{L_0}$ by \autoref{thm:nmu_eq_subatomic}. Moreover, from \autoref{ex:cyclic_unary} we know that $\ns{L_0}=\nsyn{L_0}$, and so $\ns{L}=\nsyn{L}$ as claimed.
\end{enumerate}
\section*{Details for \autoref{ex:biresid_topological}}
We prove that the map $\theta$  gives an nfa isomorphism from $N_L$ to $\rev{({N_{\rev{L}}})}$. 
Note first that if $\theta(u^{-1}L)=v^{-1}\rev{L}$, we have
   \[
     u^{-1} L \subseteq X \iff \rev{v} \in X \qquad\text{for $X\in \LW{L}$}.
    \]
In fact, 
\begin{align*}
u^{-1}L\seq X & \iff X\not\seq \tau(u^{-1}L) & \text{def. $\tau$} \\
& \iff X\not\seq \dr_L(v^{-1}\rev{L}) & \tau = \dr_L\o \theta \\
& \iff \rDR{L}(X,v^{-1}\rev{L}) & \text{by \autoref{thm:dep_thm}} \\
& \iff \rev{v}\in X & \text{def. $\rDR{L}$}
\end{align*}
With this preparation, we verify that $\theta$ satisfies the properties of an nfa morphism:

\medskip\noindent\emph{Preservation of initial and final states.}
Let $u^{-1}L\in J(\jslLQ{L})$ and $\theta(u^{-1}L)=v^{-1}\rev{L}$. Then
\[u^{-1} L \subseteq L \iff \rev{v} \in L \iff v \in \rev{L} \iff \epsilon \in v^{-1} \rev{L}.\]
A symmetric argument, exchanging the roles of $L$ and $\rev{L}$, shows that
\[  \epsilon\in u^{-1}L \iff v^{-1}\rev{L}\seq \rev{L}. \]
Thus, the state $u^{-1}L$ is initial/final in $N_L$ iff $v^{-1}\rev{L}$ is initial/final in $\rev{({N_{\rev{L}}})}$.

\medskip\noindent\emph{Preservation of transitions.}
Let $u^{-1}L, \ol{u}^{-1}L\in J(\jslLQ{L})$ and $\theta(u^{-1}L)=v^{-1}\rev{L}$, $\theta(\ol{u}^{-1}L)=\ol{v}^{-1}\rev{L}$. For each $a\in \Sigma$, we need to show that there is a transition $u^{-1}L \xra{a} \ol{u}^{-1}L$ in $N_L$ iff there is a transition $v^{-1}\rev{L}\xra{a}\ol{v}^{-1}\rev{L}$ in $\rev{({N_{\rev{L}}})}$. In fact:
    \[
      \begin{tabular}{lll}
        $\ol{u}^{-1} L \subseteq (u a)^{-1} L$
        & $\iff \rev{\ol{v}} \in (u a)^{-1} L$
        \\ & $\iff u a \rev{\ol{v}} \in L$
        \\ & $\iff \ol{v} a \rev{u} \in \rev{L}$
        \\ & $\iff {\rev{u}}\in (\ol{v} a)^{-1} \rev{L}$
        \\ & $\iff v^{-1} L^r \subseteq (\ol{v} a)^{-1} \rev{L}$.
      \end{tabular}  
    \]

\end{document}
